%% file: Manuscript.tex
\documentclass[final]{IEEEtran} 
\usepackage{fancyhdr}
\usepackage{epsfig}
\usepackage{threeparttable}
\usepackage{epsf,epsfig}
\usepackage{amsthm}
\usepackage{amsmath}
\usepackage{amssymb}
\usepackage{amsfonts}
\usepackage[noadjust]{cite}
\usepackage{dsfont}
\usepackage{subfigure}
\usepackage{color}

\include{header}

\begin{document}
\title{ Spatial Throughput of  Mobile Ad Hoc  Networks  Powered by Energy Harvesting}
\author{Kaibin Huang \thanks{K. Huang is with the Hong Kong Polytechnic University, Hong Kong.  Email: huangkb@ieee.org.  This paper has been presented in part at Asilomar Conf. on Signals, Systems, and Computers 2011 and at IEEE Intl. Conf. on Communications (ICC) 2013.  Copyright \textcopyright \ 2012 IEEE. Personal use of this material is permitted. However, permission to use this material for any other purposes must be obtained from the IEEE by sending a request to pubs-permissions@ieee.org.}}

\maketitle
\begin{abstract}
Designing  mobiles  to harvest   ambient energy  such as kinetic activities or electromagnetic radiation will   enable wireless   networks to  be self sustaining. In this paper,  the spatial throughput of a mobile ad hoc network  powered by energy harvesting  is analyzed  using a stochastic-geometry  model. In this model, transmitters are distributed  as a Poisson point process and energy arrives at each transmitter randomly with a uniform average rate called the \emph{energy arrival rate}.  Upon harvesting sufficient energy, each transmitter transmits with fixed power to an intended receiver   under   an outage-probability constraint for a target signal-to-interference-and-noise ratio. It is  assumed that transmitters  store energy in batteries with infinite  capacity. By  applying the random-walk theory, the probability that a transmitter transmits,  called the \emph{transmission probability}, is proved  to be equal to the smaller of one and the ratio between the energy-arrival rate and   transmission power.  This result  and  tools from stochastic geometry  are applied to maximize   the network throughput for a given energy-arrival rate  by optimizing transmission power. The  maximum network throughput is shown to be proportional to the optimal transmission probability, which is equal to one  if the transmitter density is below a derived function of the energy-arrival rate or  otherwise is  smaller than one and solves a given polynomial equation. 
Last, the limits of the maximum network throughput are obtained  for  the extreme cases of  high energy-arrival rates and sparse/dense networks. 

\end{abstract}


\section{Introduction}
Recent years have seen increasing  popularity of mobile devices such as sensors and smart phones, giving rise to  two design issues among others. First, the power consumption of mobile-device   networks makes an escalating contribution to global warming. Second, conventional batteries that power mobile devices periodically interrupt their operation due to finite battery lives; battery recharging or replacement is inconvenient or even impossible in certain cases.  These  issues provide  strong motivation for powering mobile devices by harvesting ambient energy such as solar energy, vibration, kinetic activities  and  electromagnetic radiation \cite{Paradiso:EnergyScavengMobile}.  The capacity of mobile-device  networks powered by energy harvesting remains largely unknown, which is addressed in this paper. 

This  paper considers a mobile ad hoc  network (MANET) where  transmitters are modeled as a homogeneous Poisson point process (PPP). Energy arrives randomly at a transmitter with a fixed average rate, called the \emph{energy-arrival rate}. The energy-arrival process is modeled as an  independent and identically distributed (i.i.d.) sequence of random variables and different processes are assumed independent.  Each transmitter deploys an energy harvester that stores arriving  energy in a rechargeable battery.  Upon harvesting sufficient energy, a transmitter   transmits with  fixed power to an intended receiver under an outage-probability constraint for a target signal-to-interference-and-noise ratio (SINR). Based on the above  model,  the network spatial throughput is maximized by optimizing  transmission power for a given energy-arrival rate.  

\subsection{Prior Work and Motivation}

The fluctuation in harvested energy due to random energy arrivals requires redesigning existing transmission algorithms for wireless  communication systems. 
Assuming infinitely backlogged data,  existing work focuses on adapting transmission power to channel states and the temporal profile of energy arrivals  to maximize the system throughput \cite{HoZhang:EnergyAllocatnEnergyHarvestConstraints, OzelUlukus:TransEnergyHarvestFading:OptimalPolicies:2011, YangUlukus:BroadcastEnergyHarvest, Zhang:MIMOBCWirelessInfoPowerTransfer}.  
For single-user systems,  the optimal power-control policies are shown to be variations of the classic water-filling policy such that the causality of energy arrivals and finite battery capacity are accounted for \cite{HoZhang:EnergyAllocatnEnergyHarvestConstraints, OzelUlukus:TransEnergyHarvestFading:OptimalPolicies:2011}.  Adaptive transmission for broadcast channels with energy harvesting has been also investigated   
\cite{YangUlukus:BroadcastEnergyHarvest, Zhang:MIMOBCWirelessInfoPowerTransfer}. In \cite{YangUlukus:BroadcastEnergyHarvest}, the optimal power-control for a two-user single-antenna broadcast channel is shown to attempt to allocate a fix amount of harvested energy  to the user with the better channel before giving the remaining energy to the other user. A two-user multiple-input-multiple-output  broadcast channel is considered in \cite{Zhang:MIMOBCWirelessInfoPowerTransfer} where one user receives data and the other scavenges transmission energy, and  the precoder at the base station is designed to optimize the  tradeoff between the data rate and the rate of  harvested energy. 

{ 
In wireless communication systems with both bursty data-and-energy arrivals, buffering energy and data creates two corresponding queues at each transmitter. Jointly controlling these two coupled  queues  is more challenging than controlling only the data queue in traditional systems with reliable power supplies \cite{GeorgNeelyBook}. The algorithms  for optimally  controlling the energy-and-data queues have  been proposed  for single-user systems \cite{YangUlukus:PacketScheduleBroadcastChannelEnergyHarvest} and downlink systems  \cite{Uysal:PacketScheduleEnergyHarvestBC:2011} to minimize the packet transmission delay,  for interference channels to minimize queueing delay \cite{Lau:DelayOptimInterfNet:2012}, and for downlink systems to maximize the system throughput \cite{Gatzianas:ControlNetworkRechargeableBatt:2010}. These algorithms  
share a common objective of optimizing a particular performance metric for given average harvested power. The objective is aligned with that for designing the traditional energy-efficient systems with only data queues, namely minimizing the average transmission power under a performance constraint such as  fixed packet-transmission  delay  for single-user systems \cite{Uysal:EnergyEfficientPacketTX:2002, Neely:EnergyEfficientDelayConstraints:2008}, allowed queueing delay for downlink systems  \cite{Neely:EnergyDelayTradeoffMultiuserDL:2007}, and given traffic in wireless networks \cite{Neely:EnergyControlTimeVarNet:2006}. }

Wireless networks with energy harvesting have been studied  \cite{Simeone:MACSensorNetEnergyHarvest:2012,  Kansal:EnergyHarvestingSensorNets:2006, Ephremides:StabilityMACEnergyHarvesting:2011}.  For a wireless  sensor network with energy harvesting   and based on a simple channel model that omits channel noise and path loss,  the probability that a sensor  successfully transmits a data packet  to a fusion center is analyzed  in \cite{Simeone:MACSensorNetEnergyHarvest:2012} for different multiple-access protocols including time-division multiple access and Aloha like random access. Managing traffic load  in time and space is important for wireless sensor networks to be self sustaining through energy harvesting. Therefore, distributive strategies are proposed in  \cite{Kansal:EnergyHarvestingSensorNets:2006} for adapting traffic load to the spatial-and-temporal energy profile and evaluated using a network prototype. For a two-user interference network with energy harvesting, the data-and-energy arrivals are modeled as Bernoulli  processes and the stability region is characterized  such that it comprises all data-rate pairs under the constraint of finite data-queue lengths \cite{Ephremides:StabilityMACEnergyHarvesting:2011}. In view of prior work, there are few results that quantify   the  tradeoff  between  the network throughput and the energy-arrival rate though  such results  specify   the fundamental limit of the  network performance. This   tradeoff  is  investigated in the sequel using a stochastic-geometry approach. 

Stochastic geometry  provides a set of  powerful mathematical tools for modeling and designing wireless networks \cite{HaenggiAndrews:StochasticGeometryRandomGraphWirelessNetworks}.  
MANETs based on random access and carrier-sensing multiple access have been modeled using the PPPs \cite{Baccelli:AlohaProtocolMultihopMANET:2006, WeberAndrews:TransCapWlssAdHocNetwkOutage:2005} and Matern hard-core processes \cite{Baccelli:StochasticGeometryWLAN:2007}, respectively. 
Cellular networks have been shown to be suitably modeled using the Poisson Voronoi tessellation \cite{Andrews:TractableApproachCoverageCellular:2010}. Models of coexisting networks can be constructed by superimposing multiple point processes \cite{Dhillon:HeterNetModelAnalysis:2012,Huang:SpectrumSharCellularAdHocNetTxCapTradeoff:2008}. Stochastic-geometry models of wireless networks have been employed to quantify the network-performance gains due to   physical-layer techniques such as opportunistic transmission
\cite{WeberAndrews:TransCapAdHocNetwkDistSch:2006}, bandwidth partitioning \cite{JindalAndrews:BandwidthPartitioning:2007}, successive interference cancellation \cite{WeberAndrews:TransCapWlssAdHocNetwkSIC:2005}, and multi-antenna techniques \cite{AndrewJeff:CapacityScalingSpatialDiversity:2006,  VazeHeath:TransCapacityMultipleAntennaAdHocNetwork, Jindal:RethinkMIMONetwork:LinearThroughput:2008, LouieMacKay:SpatialMultiplexDiversityAdHocNetwork, Huang:SpatialInterfCancel:2012}. The performance metric typically considered in the literature is the network spatial throughput under a constraint on the outage probability for a target  SINR, which is also adopted in this paper. Using this metric, most prior work focuses on  deriving  the  outage probability using techniques   such as the Laplace transform \cite{Baccelli:AlohaProtocolMultihopMANET:2006, Andrews:TractableApproachCoverageCellular:2010} and probabilistic inequalities \cite{WeberAndrews:TransCapWlssAdHocNetwkOutage:2005, WeberAndrews:TransCapAdHocNetwkDistSch:2006}. 
This paper considers a MANET with Poisson distributed transmitters similar to the existing literature (see e.g., \cite{WeberAndrews:TransCapWlssAdHocNetwkOutage:2005}). However, the transmitters in the current network model are powered by energy harvesting  instead of reliable power supplies as in prior work. The consideration of energy harvesting  introduces several new design issues including the aforementioned tradeoff  between the  network throughput and energy-arrival rate, the corresponding optimization of transmission power, and the effect of finite energy storage, which are investigated  in the sequel. 

\subsection{Contributions and Organization}

For exposition, a few definitions and notations are provided as follows. Time is slotted. Define the transmission probability $\rho$ as the probability that a transmitter transmits and the   \emph{network interference temperature} as the maximum active transmitter density under the outage-probability constraint. Let $\lambda_0$ denote the transmitter density,    $\lambda_e$  the energy-arrival rate,  $P$ the transmission power, and $Z$  a nonnegative random variable representing the amount of energy harvested  by a typical  harvester in an arbitrary   slot.\footnote{A typical point  is  selected from a spatial point process by uniform sampling.} { Note  that $\E[Z] = \lambda_e$ and the density of active transmitters is equal to  $\rho \lambda_0$.}

The  main contributions of this paper are summarized as follows. 

\begin{enumerate}
\item  Assume  infinite battery capacity. Using  the law of large numbers and  random-walk theory, it is  proved  that $\rho$ is equal to the smaller 
of $\lambda_e/P$ and one. { It is worth mentioning that the tractable analysis relies on assuming the sub-optimal fixed-power transmission.  To the best of the author's knowledge, the aforementioned result is unknown from existing work that mostly focuses on  designing the optimal adaptive-transmission algorithms \cite{HoZhang:EnergyAllocatnEnergyHarvestConstraints, OzelUlukus:TransEnergyHarvestFading:OptimalPolicies:2011, YangUlukus:BroadcastEnergyHarvest, Zhang:MIMOBCWirelessInfoPowerTransfer, YangUlukus:PacketScheduleBroadcastChannelEnergyHarvest, Uysal:PacketScheduleEnergyHarvestBC:2011, Lau:DelayOptimInterfNet:2012, Gatzianas:ControlNetworkRechargeableBatt:2010}.}

\item Consider the case of finite battery capacity.    Bounds on $\rho$ are derived, which converge to the results stated above as the battery capacity increases. Moreover, two special cases are considered. If $Z$ is bounded and no larger than $P$, it is shown that   $\rho$ is equal to $\lambda_e/P$  so long as the battery capacity is larger than $2P$. If $Z$ is a discrete random variable, $\rho$ is  analyzed   using Markov-chain theory.

\item Assume infinite battery capacity.    By  applying derived results on transmission probability and tools from stochastic geometry,    the network throughput is  maximized by   optimizing $P$ for   given  $\lambda_e$. Consider the condition that  $\lambda_0$ is smaller than the network interference temperature evaluated for equal $P$ and $\lambda_e$. If this condition holds, the maximum throughput $R^*$ is shown to be 
\begin{equation}
R^* = \lambda_0  \log_2(1+\theta)\nn
\end{equation}
where $\theta$ is the target SINR. If the aforementioned condition is not satisfied, 
\begin{equation}
R^* = \frac{\lambda_0\lambda_e}{P^*}  \log_2(1+\theta)\nn
\end{equation}
where the optimal transmission power $P^*$ is larger than $\lambda_e$ and solves a derived  polynomial equation.

\item Furthermore, the limits of the maximum network throughput are  obtained for the extreme cases of high energy-arrival rates ($\lambda_e \rightarrow\infty$) and   dense   networks ($\lambda_0 \rightarrow\infty$). Specifically. 
\begin{align}
\lim_{\lambda_e\rightarrow\infty}R^*(\lambda_e) &= \min\l(\lambda_0, \frac{\mu_{\epsilon}}{\theta^{\frac{2}{\alpha}}}\r)\log_2(1+\theta)\nn\\ 
\lim_{\lambda_0\rightarrow\infty}R^*(\lambda_0) &=  \frac{\mu_{\epsilon}}{\theta^{\frac{2}{\alpha}}}\log_2(1+\theta) \nn
\end{align}
where $\mu_{\epsilon}$ is a positive constant determined by the maximum outage probability. 
\end{enumerate}

The remainder of this paper is organized as follows. The network model and performance metric  are described in Section~\ref{Section:Models}. The transmission probability is  analyzed in Section~\ref{Section:TXProb}. The results are applied to maximize the network throughput in Section~\ref{Section:Thput}. Numerical  results are presented in Section~\ref{Section:Simulation} followed by concluding remarks in Section~\ref{Section:Conclusion}.

\section{Model and Metric }\label{Section:Models}

\subsection{Network Model}
As illustrated in Fig.~\ref{Fig:Sys},   the transmitters $\{T\}$ of the MANET are  distributed in the Euclidean plane $\mathds{R}^2$ following  a homogeneous PPP $\Phi$ with density $\lambda_0$, where $T$ denotes the coordinates of a transmitter.  Each transmitter is associated with an intended receiver located at a unit distance,  {which is assumed to simplify  the expression for the received signal power by omitting the data-link path loss.} The signal transmitted  by $T$ with power $P$ is received by a receiver located at $X$ with power equal to $P|X-T|^{-\alpha}$ with $\alpha  > 2$ being the path-loss exponent. In other words, propagation is characterized by path loss while  fading is omitted to simplify notation.\footnote{The consideration of random transmission distances and fading has no effect on the main  results except that the parameter $\mu_\epsilon$ defined in \eqref{Eq:NomDen} has to be redefined by including additional random variables.}

\begin{figure}[t]
\begin{center}
\includegraphics[width=8.5cm]{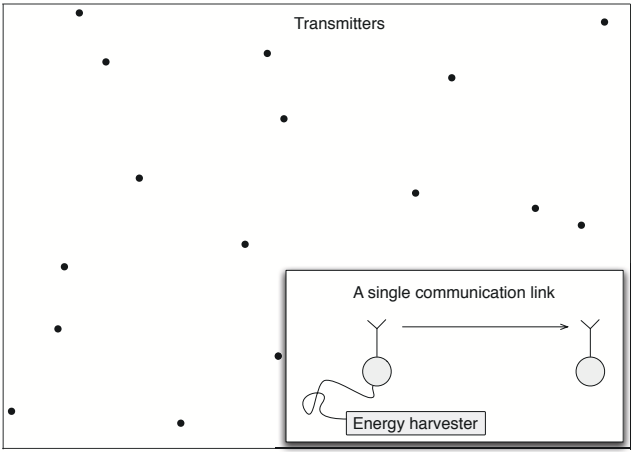}
\caption{Single-antenna transmitters in the MANET are modeled as a homogeneous PPP in  the horizontal plane. Each transmitter is powered by an energy harvester and transmits to an intended receiver at an unit distance.   }
\label{Fig:Sys}
\end{center}
\end{figure}

Time is partitioned into slots of  unit duration  with $t$ denoting  the slot index. The amount of energy harvested by the typical harvester in the $t$-th slot is represented by  the nonnegative random variable $Z_t$.  

\begin{assumption}\label{AS:Z:Dist}\emph{The energy-arrival process $\{Z_t\}\subset\mathds{R}^+$ is an i.i.d. sequence and independent of other  energy-arrival processes. Moreover,  the \emph{cumulant  generating function} of the random variable $(Z_t - \beta)$ with $\beta$ being a given constant, namely  $\ln \E\l[e^{r(Z_t- \beta)}\r]$,  has a  root $r^*(\beta)$ such that $r^*(\beta) > 0 $ if  $\beta >  \lambda_e$ and $r^*(\beta) < 0 $ if  $\beta <  \lambda_e$.}
\end{assumption}

This assumption allows  the use of  results on the large deviation of random walks in the  subsequent analysis \cite{GallagerBook:StochasticProcs:95}. Let $B$ denote the battery capacity identical for all harvesters. Moreover, the  typical  transmitter and the battery level of the corresponding (typical) harvester are represented by $T_0$ and $S_t$, respectively.  A transmitter transmits one  data packet with fixed power $P$ whenever the corresponding battery level exceeds $P$.  As a result, $S_t$ evolves as 
\begin{equation}\label{Eq:Battery:Evol}
S_t = \min(S_{t-1} + Z_{t}  - P I(S_{t-1}  \geq P), B), \quad t = 1, 2, \cdots
\end{equation}
where $S_0 = 0$ and the indicator function $I(A)$ for an event $A$ is equal to one if $A$ occurs  or else is zero. The battery-level evolutions in  prior work are similar to that in \eqref{Eq:Battery:Evol} but with fixed power $P$ replaced with power adapted to factors such as the channel state and battery level \cite{HoZhang:EnergyAllocatnEnergyHarvestConstraints,  Gatzianas:ControlNetworkRechargeableBatt:2010, Lau:DelayOptimInterfNet:2012}.

\subsection{Performance Metric}
Assume infinitely backlogged and packetized data.   The transmission probability $\rho$ can be written as 
\begin{equation}\label{Eq:TXProb:Def}
\rho = \lim_{n\rightarrow\infty}\frac{1}{n}\sum_{t=1}^n \E[I(S_t \geq P)].
\end{equation}
 According to Coloring  Theorem \cite{Kingman93:PoissonProc},  the process of active  transmitters, denoted as  $\Pi$,  is a PPP   with density $\lambda_t = \rho \lambda_0$. Data  is encoded at a fixed rate $\log_2(1+\theta)$ bit/s/Hz with  $\theta > 0$ being the target SINR.  Correct decoding of a data packet  requires the received SINR to be no smaller than $\theta$ or else an outage event occurs. The   outage probability $\Pout$ is defined as $\Pout = \Pr(\SINR < \theta)$ where $\SINR$ represents the received SINR at the  receiver for $T_0$. It is assumed that 
 the  receiver for $T_0$ is  located at the origin, which  does not compromise the  generality based on Slyvnyak's Theorem \cite{StoyanBook:StochasticGeometry:95}, and that noise has unit  variance. Based on these assumptions,  $\Pout$ can be written   as
\begin{align}
\Pout  &= \Pr\l(\frac{P}{\sum\limits_{T\in\Pi\backslash\{T_0\}}P|T|^{-\alpha}+1} < \theta\r)\label{Eq:Pout:aa}\\
&=\Pr\l(\sum_{T\in\Pi\backslash\{T_0\}}|T|^{-\alpha} > \frac{1}{\theta} - \frac{1}{P}\r)\label{Eq:Pout:aaa}\\
&=\Pr\l(\sum_{T\in\Pi}|T|^{-\alpha} > \frac{1}{\theta}-\frac{1}{P} \r)  \label{Eq:Pout:a}
\end{align}
where the summation in \eqref{Eq:Pout:aa} represents the interference power and \eqref{Eq:Pout:a} uses  Slyvnyak's Theorem.  It is worth mentioning that the probabilities  in \eqref{Eq:Pout:aa} and \eqref{Eq:Pout:aaa} are \emph{palm measures}  \cite{StoyanBook:StochasticGeometry:95} but that in \eqref{Eq:Pout:a} is not. To ensure the quality-of-service, an  outage-probability constraint is applied such that $\Pout \leq \epsilon$ with $0 < \epsilon \ll 1$.  The performance metric  is the spatial network-throughput density $R$ (bit/s/Hz/unit-area) that is referred to simply as the network throughput and defined as
\begin{align}
R &= \lambda_t  \log_2(1+\theta) \label{Eq:Thput:Def}\\
 &= \lambda_0\rho  \log_2(1+\theta) \label{Eq:Thput:Def:a}
\end{align}
where $\rho$  is controlled by adjusting  $P$ such that    the outage-probability constraint is satisfied. To be precise, $R$ should be scaled by the success  probability $(1-\Pout)$ but this factor  is close to one given $\epsilon \ll 1$  and thus  omitted for ease of notation.

\section{Transmission Probability}\label{Section:TXProb}

\subsection{Infinite Battery Capacity}
Deriving transmission probability requires  analyzing the distribution of battery levels at  energy harvesters. 
By substituting $B\rightarrow\infty$ into \eqref{Eq:Battery:Evol}, the battery level at the typical energy harvester with infinite battery capacity evolves  as 
\begin{equation}\label{Eq:Battery:Evol:Infinite}
S_{t} = S_{t-1} + Z_t  - P I(S_{t-1}  \geq P).
\end{equation}
The distribution of $S_t$ can be related to  the threshold-crossing probability for a random walk as follows. 
Denote the instants when the battery level crosses the threshold $P$ from below as $t_1, t_2, \cdots$, namely that   $S_{t_{n}-1} < P$ and $S_{t_{n}} \geq  P$ for $n=1, 2, \cdots$.  These time instants are grouped into the set $\mathcal{T} = \{t_1, t_2, \cdots\}$. 
Moreover, define the random variable $\bar{Z}_t = Z_t - P$ and two random processes $\{G_t\}$ and $\{G'_t\}$ as 
\begin{align}
G_{t} &= \max(G_{t-1} + \bar{Z}_t, 0)\label{Eq:RandProc1:Def}\\
G_t' &= \l\{\begin{aligned}
&S_t, && t \in \mathcal{T}\\
&G'_{t-1}, && t \notin \mathcal{T} 
\end{aligned}
\r.\label{Eq:RandProc2:Def}
\end{align}
with $G_0 =0 $ and $G_0'= P$. 
Based on \eqref{Eq:RandProc1:Def}, the  probability that  $\{G_t\}$ crosses a threshold $x > 0$ in the $t$-th slot  can be written as 
\begin{equation}
\begin{aligned}
 \Pr(G_t > x) = &\\ \Pr\Bigg(\max\Bigg(&\bar{Z}_t, \bar{Z}_t + \bar{Z}_{t-1}, \cdots, \sum_{n=1}^t \bar{Z}_n\Bigg) > x\Bigg). 
 \end{aligned} \label{Eq:Prob:G}
\end{equation}
Consider the random walk $\{\sum_{n=1}^m \bar{Z}_{t - m+1}\}$ starting in the $t$-th slot and progressing backwards. The probability in \eqref{Eq:Prob:G} can be interpreted as the probability  that the said  random walk with a negative drift ever crosses the threshold $x$ by the $t$-th step. Applying Kingman bound on the threshold-crossing probability for a random walk \cite[p234]{GallagerBook:StochasticProcs:95} gives that for all $t > 0$, 
\begin{equation}\label{Eq:Gt:Cross}
\Pr(G_t > x) \leq e^{-r^*(P) x}, \qquad \lambda_e < P
\end{equation}
where $r^*(P)$ as defined in Assumption~\ref{AS:Z:Dist} with $\beta = P$ is  the positive root of the cumulant generating function of $\bar{Z}_t$.

\begin{lemma}\label{Lem:St:UB}\emph{Given infinite battery capacity, the battery level $S_t$ satisfies 
\begin{equation}
S_t  \leq G_t + G'_t. \nn
\end{equation}
}
\end{lemma}
The proof of Lemma~\ref{Lem:St:UB} is provided in Appendix~\ref{App:St:UB}. Using \eqref{Eq:Gt:Cross} and Lemma~\ref{Lem:St:UB}, the threshold-crossing probability for the battery level can be shown to be bounded as follows. 

\begin{lemma}\label{Lem:St:Cross}\emph{Given infinite battery capacity and $ \lambda_e < P$, the distribution of  the battery level $S_t$  satisfies 
\begin{equation}\label{Eq:St:Cross}
\Pr(S_t > x) \leq 2e^{-\frac{1}{2}r^*(P)(x - 2P)} 
\end{equation}
with $r^*(P) > 0$. 
}
\end{lemma}

The proof of Lemma~\ref{Lem:St:Cross} is given in Appendix~\ref{App:St:Cross}. Define the \emph{energy-overshoot function} $D_t: \mathds{R}^+ \rightarrow\mathds{R}^+$ as the expected amount of energy stored in the typical  harvester  in excess of a threshold $x > 0 $ in the $t$-th slot: 
\begin{equation}\label{Eq:EnergEx} 
D_t(x) = \int_x^\infty (y-x) f_{s}(y, t) dy 
\end{equation}
where $f_s(y, t)$ represents the probability density function  of $S_t$. The function $D_t(x)$ can be bounded as shown in Lemma~\ref{Lem:ExcessEnergy}, which is proved in Appendix~\ref{App:ExcessEnergy} using Lemma~\ref{Lem:St:Cross}.

\begin{lemma}\label{Lem:ExcessEnergy}\emph{Given infinite battery capacity and $\lambda_e < P$, the energy-overshoot function $D_t(x)$ satisfies 
\begin{equation}
D_t(x)  \leq \frac{4}{r^*(P)} e^{-\frac{1}{2}r^*(P)(x - 2P)}, \qquad \forall \ t \geq 0 \nn
\end{equation}
with $r^*(P) > 0$. 
}
\end{lemma}

Using  Lemma~\ref{Lem:ExcessEnergy}, the main result of this section is readily obtained as shown below. 

\begin{theorem}\label{Theo:TXProb:InfBattery} \emph{Given infinite battery capacity,  the transmission  probability  is 
\begin{equation}
\rho =\min\l(1, \frac{\lambda_e}{P}\r).\nn
\end{equation}}
\end{theorem}
\begin{proof}
First, consider the case of  $ \lambda_e > P$. 
Replacing the indicator function  in \eqref{Eq:Battery:Evol:Infinite} with one yields a lower bound on $S_t$, namely that $S_t \geq \sum_{m=1}^t \bar{Z}_m$. As a result,   $\rho$ given in \eqref{Eq:TXProb:Def} can be lower bounded as 
\begin{align}
\!\!\!\!\rho & \geq\! \lim_{n\rightarrow\infty} \frac{1}{n} \sum_{t=1}^n \E\l[I\l( \sum_{m=1}^t\bar{Z}_m \geq P \r)\r]\nn\\
& =\! \lim_{n\rightarrow\infty} \frac{1}{n} \sum_{t=1}^n \Pr\l(\!\frac{1}{t}\!\sum_{m=1}^t\! (Z_m\!-\!\lambda_e)\!  \geq\! P\! -\! \lambda_e \!+\! \frac{P}{t}\!\r). \label{Eq:TXProb:LB}
\end{align}
Using $\E[Z_n] = \lambda_e$ and  applying the  weak law of large numbers  \cite{GallagerBook:StochasticProcs:95},  for given $\tau > 0$ and $\delta > 0$, there exists $k > 0$ such that for all $t \geq k$, 
\begin{equation}\label{Eq:LLN}
\Pr\l(\frac{1}{t}\sum_{m=1}^t (Z_m - \lambda_e) \geq -\tau \r) \geq 1 - \delta. 
\end{equation}
Since $\lambda_e > P$ , it follows that given $\delta > 0$, there exist $\tau > 0$ and $k' > 0$ such that $\tau <  (\lambda_e - P - \frac{P }{t})$ for all $t \geq k'$.  
Using this fact and  substituting \eqref{Eq:LLN} into \eqref{Eq:TXProb:LB}, 
\begin{align}
\rho & \geq \lim_{n\rightarrow\infty} \frac{1}{n} \sum_{t=\max(k, k')}^n \Pr\l(\frac{1}{t}\sum_{m=1}^t (Z_m - \lambda_e) \geq -\tau\r) \nn\\
&\geq \lim_{n\rightarrow\infty} \frac{1}{n} \sum_{t=\max(k, k')}^n (1-\delta) \nn\\
& = \lim_{n\rightarrow\infty} \frac{[n - \max(k, k')](1-\delta)}{n} \nn\\
& = 1 - \delta. \label{Eq:OneDelta:a}
\end{align}
As $\delta$ is arbitrary and $\rho \leq 1$,  the desired result for the case of $\lambda_e > P$ follows from \eqref{Eq:OneDelta:a} and letting $\delta \rightarrow 0$.

Next, consider the case of $ \lambda_e < P$.   The  expected total  amounts of  harvested and transmitted energy   by the $t$-th slot differ by the battery level in the $t$-th  slot, namely 
\begin{equation}\label{Eq:Total:Energy}
\sum_{m=1}^t \E[Z_m] = P \sum_{m=1}^t \E[I(S_{m-1} > P)] + \E[S_t]. 
\end{equation}
Since $\E[S_t] \geq 0$, 
\begin{align}
\lambda_e &\geq P\lim_{t\rightarrow\infty} \frac{1}{t}\sum_{m=1}^t \E[I(S_{m-1} > P)]\nn\\
& = P\rho\label{Eq:PRo}
\end{align}
where \eqref{Eq:PRo} is obtained using the definition of $\rho$ in   \eqref{Eq:TXProb:Def}. 
It follows that 
\begin{equation}
\rho \leq \frac{\lambda_e}{P}, \qquad \lambda_e < P. \label{Eq:TxProb:UB}
\end{equation}
Note that $\E[S_t] \leq x + D_t(x)$ with $x > 0$. 
Using Lemma~\ref{Lem:ExcessEnergy}, for given $\delta > 0$, there exists $x > 0$ such that $\E[S_t] \leq x + \delta$. Combining this fact and \eqref{Eq:Total:Energy} yields 
\begin{align}
\lambda_e &=  P\lim_{t\rightarrow\infty} \frac{1}{t}\sum_{m=1}^t \E[I(S_{t-1} > P)] + \lim_{t\rightarrow\infty}\frac{\E[S_t]}{t} \nn\\
&\leq P\lim_{t\rightarrow\infty} \frac{1}{t}\sum_{m=1}^t \E[I(S_{t-1} > P)] + \lim_{t\rightarrow\infty}\frac{x + \delta}{t} \nn\\
& = P\rho.  \nn
\end{align}
As a result, 
\begin{equation}
\rho \geq \frac{\lambda_e}{P}, \qquad \lambda_e <  P. \label{Eq:TxProb:LB}
\end{equation}
The desired result for the case of $\lambda_e < P$ is proved by combining \eqref{Eq:TxProb:UB} and \eqref{Eq:TxProb:LB}. 

Last, the desired result  for the boundary case of $\lambda_e = P$ is proved by using the results proved above for $\lambda_e \neq P$ and letting $P \rightarrow\lambda_e$ from either the right or the left, completing the proof. 
\end{proof}

\begin{remark}\emph{According to Theorem~\ref{Theo:TXProb:InfBattery}, if  $P  < \lambda_e$, each transmitter  transmits continuously in the steady state and is free of interruption caused by energy shortage. However, continuous transmissions are at the cost that the fraction of harvested energy at the rate of  $(\lambda_e - P)$ is never used for transmission and hence wasted. Next, if $P  >  \lambda_e$,  there exists nonzero probability that the battery level of a transmitter is below $P$ and hence transmission can be  interrupted. Nevertheless, all harvested energy will be eventually used for transmission.} 
\end{remark}

\subsection{Finite Battery Capacity} \label{Section:Finite:B}The dynamics of the battery level $S_t$  at the typical harvester   are characterized in \eqref{Eq:Battery:Evol}. Let $\tilde{D}_t(x)$ denote the energy-overshoot function for the case of finite-battery capacity, which is defined similarly as $D_t(x)$ in the preceding section. Given finite battery capacity, battery overflow can occur such that the battery saturates  and arriving energy has to be discarded, where the expected amount of discarded energy is measured by $\tilde{D}_t(B)$. For the case of $\lambda_e < P$,  $\tilde{D}_t(B)$ can be bounded by the same upper bound on $D_t(x)$ (see Lemma~\ref{Lem:ExcessEnergy}) as shown below. 

\begin{lemma}\label{Lem:ExcessEnergy:B}\emph{Given finite battery capacity and $\lambda_e < P$, the energy-overshoot function $\tilde{D}_t(B)$ satisfies 
\begin{equation}
\tilde{D}_t(B)  \leq \frac{4}{r^*(P)} e^{-\frac{1}{2}r^*(P)(B - 2P)} \nn
\end{equation}
with  $r^*(P) > 0$. 
}
\end{lemma}
Lemma~\ref{Lem:ExcessEnergy:B} is proved in Appendix~ \ref{App:ExcessEnergy:B}. { Next, the tail probability of $S_t$ can be bounded for the case of $\lambda_e > P$ as shown in the following lemma, which is  proved in Appendix~\ref{App:BatteryTail:B}. 

\begin{lemma}\label{Lem:BatteryTail:B}\emph{Given finite battery capacity and $\lambda_e > P$, the distribution of the battery level $S_t$ satisfies 
\begin{equation}
\lim_{n\rightarrow\infty}\frac{1}{n}\sum_{t=1}^n\Pr(S_t < x)\leq e^{r^*(P)(B-x)} \nn
\end{equation}
with  $r^*(P) < 0$ and $x\in[0, B]$. 
}
\end{lemma}}

Using Lemma~\ref{Lem:ExcessEnergy:B} and \ref{Lem:BatteryTail:B}, bounds on the transmission probability are obtained as follows.

\begin{proposition}\label{Prop:TXProb:FiniteBattery} \emph{Given  finite battery capacity,   the transmission probability  $\rho$ satisfies the following. 
\begin{enumerate}
\item If $\lambda_e < P$
\begin{equation}
\frac{\lambda_e}{P}\l[1 - \frac{4}{\lambda_e r^*(P)} e^{-\frac{1}{2}r^*(P)(B - 2P)}\r] \leq \rho \leq \frac{\lambda_e}{P} \label{Eq:TXProb:B:a}
\end{equation}
with    $r^*(P) > 0$. 
\item If $\lambda_e > P$  
\begin{equation}
1 - e^{r^*(P)(B-P)}\leq \rho \leq 1 \label{Eq:TXProb:B:b}
\end{equation}
with $r^*(P) <  0$. 
\item If $\lambda_e = P$  
\begin{equation}
\!\!\!\!
\begin{aligned}
&\max_{x > 0}\l\{\frac{\lambda_e}{\lambda_e + x}\l[1 - \frac{4e^{-\frac{1}{2}r^*(\lambda_e+x)(B - 2\lambda_e-2x)}}{\lambda_e r^*(\lambda_e+x)} \r]\r\}  \\
&\qquad \leq \rho\leq 1 
\end{aligned}\label{Eq:TXProb:B:c}
\end{equation}
with $r^*(\lambda_e+x) >  0$ given $x>0$.
\end{enumerate}
}
\end{proposition}
\begin{proof} Consider the case of $\lambda_e < P$. By accounting for the discarded energy due to battery overflow, the expected total  amounts of transmitted and harvested energy by the $t$-th slot is  related by modifying \eqref{Eq:Total:Energy} as 
\begin{equation}
\sum_{m=1}^t \E[Z_m] = P \sum_{m=1}^t \E[I(S_{m-1} > P)] + \E[S_t] + \sum_{m=1}^t \tilde{D}_m(B) \nn
\end{equation}
where the last term  gives  the expected amount of discarded energy.   Applying Lemma~\ref{Lem:ExcessEnergy:B} and $S_t \leq B$ yields 
\begin{align}
\lambda_e &\leq P \lim_{t\rightarrow\infty}\frac{1}{t}\sum_{m=1}^t \E[I(S_{m-1} > P)] + \lim_{t\rightarrow\infty}\frac{B}{t} + \nn\\
&\qquad  \frac{4}{r^*(P)} e^{-\frac{1}{2}r^*(P)(B - 2P)}\nn\\ 
&\leq P \rho + \frac{4}{r^*(P)} e^{-\frac{1}{2}r^*(P)(B - 2P)}\nn
\end{align}
and the first inequality in \eqref{Eq:TXProb:B:a} follows. The second inequality is proved using Theorem~\ref{Theo:TXProb:InfBattery} and the fact that limiting the battery capacity reduces $\rho$.

Next, consider the case of  $\lambda_e > P$. The definition of $\rho$ in \eqref{Eq:TXProb:Def} can be rewritten as 
\begin{equation}\label{Eq:TXProb:Def:Alt}
\rho = 1-\lim_{n\rightarrow\infty}\frac{1}{n}\sum_{t=1}^n\Pr(S_t < P). 
\end{equation}
The last term can be upper bounded using Lemma~\ref{Lem:BatteryTail:B}, yielding the first inequality in \eqref{Eq:TXProb:B:b}. The second inequality is trivial.

Last, consider the case of  $\lambda_e = P$. Let $\rho'(x)$ denote the transmission probability for the virtual scenario where all transmissions use the  power  $(P + x)$ with $x > 0$. It can be proved similarly as the first inequality in \eqref{Eq:TXProb:B:a} that 
\begin{equation}
\rho'(x) \geq \frac{\lambda_e}{\lambda_e + x}\l[1 - \frac{4e^{-\frac{1}{2}r^*(\lambda_e+x)(B - 2\lambda_e-2x)}}{\lambda_e r^*(\lambda_e+x)} \r] \label{Eq:TXProb:B:d}
\end{equation}
with $x > 0$.
Note that $\rho'(x)$ also gives the transmission probability for 
a virtual transmission strategy that removes  $x$ unit of energy from the battery of a harvester following every instance of transmission.  Since removing energy from the battery reduces the transmission probability, $\rho \geq \rho'(x)$ holds. Combining this inequality  and \eqref{Eq:TXProb:B:d} yields 
\begin{equation}
\rho \geq \frac{\lambda_e}{\lambda_e + x}\l[1 - \frac{4}{\lambda_e r^*(\lambda_e+x)} e^{-\frac{1}{2}r^*(\lambda_e+x)(B - 2\lambda_e-2x)}\r]\nn
\end{equation}
for all $\ x > 0$.
The first inequality in \eqref{Eq:TXProb:B:b} follows. The second inequality is trivial, completing the proof.
\end{proof}

\begin{remark}\emph{
For a sanity check, it can be observed from Proposition~\ref{Prop:TXProb:FiniteBattery} that as $B\rightarrow\infty$,  $\rho$ converges to its counterpart for  the case of infinite battery capacity as stated in Theorem~\ref{Theo:TXProb:InfBattery}.}
\end{remark}

\begin{remark}\emph{
By comparing Propostion~\ref{Prop:TXProb:FiniteBattery} with Theorem~\ref{Theo:TXProb:InfBattery}, it is observed that the degradation  of $\rho$ due to finite battery capacity decreases exponentially with increasing $B$. Hence the effect of finite  battery capacity on $\rho$ is expected to diminish rapidly as $B$ increases,  which is confirmed by simulation results in the sequel.}
\end{remark}

\begin{remark}\emph{The battery-level process for the case of $\lambda_e = P$ is related to  a random walk with a zero drift for which $r^*(\lambda_e)$ is not defined and the   threshold-crossing probability does not have an exponential upper bound \cite{GallagerBook:StochasticProcs:95}. This causes the difficulty  in deriving a lower bound on $\rho$ simpler than that in \eqref{Eq:TXProb:B:c}. Moreover, the maximization of the lower bound in \eqref{Eq:TXProb:B:c} cannot be solved analytically. One should not expect that the bound  is maximized  as $x \rightarrow 0$ because the function $r^*(\lambda_e + x)$ can be a monotone increasing function of $x$.  For instance, given that $Z_t$ follows the exponential distribution with unit mean, it is obtained that 
\begin{equation}
r^*(1+x) = \frac{W_0\l(-(1+x)e^{-(1+x)}\r)}{1+x} + 1 \nn
\end{equation}
where $W_0$ denotes the $0$-th branch of the Lambert W function. The function can be plotted and observed to be monotone increasing for $x \geq 0$.}
\end{remark}
In general, exact analysis of $\rho$ for the case of finite battery capacity is challenging except for some special cases, two of which are discussed as follows. 

\subsubsection{Special case: bounded energy arrivals} Consider the case that  $Z_t$ has bounded support and    $Z_t \in [0, z_{\max}]$. 

\begin{proposition} \label{Eq:Bounded:Energy}\emph{Consider bounded energy arrivals.   If $z_{\max}\leq P$ and the battery capacity $B > 2P$, the probability for battery-overflow is zero and the transmission probability is 
\begin{equation}\label{Eq:TXProb:Bounded}
\rho = \frac{\lambda_e}{P} 
\end{equation}
where $\lambda_e \leq P$.}
\end{proposition}
\begin{proof} By expanding  \eqref{Eq:Battery:Evol}, 
\begin{equation}\label{Eq:Evol:Bounded}
S_t =\l\{ \begin{aligned}
&\min(S_{t-1} + Z_t  - P, B), && S_{t-1} \geq P\\
&\min(S_{t-1} + Z_t, B), && S_{t-1} <  P. 
\end{aligned}\r.
\end{equation}
Given  $B > 2P$ and $Z_t \leq P$, it follows from \eqref{Eq:Evol:Bounded} that    $S_{t-1} \leq 2P$ ensures  $S_{t} \leq 2P$. 
Since $S_0 = 0$, this result implies zero battery-overflow probability. Consequently, the transmission probability is identical to that for the case of  infinite battery capacity in Theorem~\ref{Theo:TXProb:InfBattery}, proving the desired result in \eqref{Eq:TXProb:Bounded}. 
\end{proof}

\subsubsection{Special case: discrete energy arrivals} Assume that $Z_t$ is a discrete random variable and takes on values from $\{0, 1, 2, \cdots\}$, and that $P$ and $B$ are  positive integers with $B \geq P$. Under these assumptions, the distribution of battery levels can be analyzed using Markov-chain theory. Since $S_t$ is independent of $\{S_n\}_{n=1}^{t-2}$ given $S_{t-1}$, $\{S_t\}$ satisfies the Markov property and  is hence a Markov chain. Given battery capacity $B$, the Markov chain $\{S_t\}$  has the state space $\{0, 1, \cdots, B\}$. Define the transition probability $p_{mn}$ as $p_{mn} = \Pr(S_t = n\mid S_{t-1} = m)$. { If $n < B$, the battery level is below its limit and the transition probability is given as 
\begin{equation}
p_{mn}  = \Pr\l( Z_t = n - m + PI(m \geq P)\r), \qquad n < B\label{Eq:TXProb:a}
\end{equation}
where the indicator function $I$ specifies if energy of $P$ units is consumed for transmission in the current slot depending on if $m$ reaches $P$.   If $n = B$, the transmission from state $m$ to $n$ includes all events that the energy arrival in the current slot causes battery saturation, namely $Z_t = \{B - m, B - m +1, \cdots\}$ if $m < P$ and  $Z_t = \{B - m + P, B - m + P+1, \cdots\}$ if $m \geq P$. It follows that 
\begin{equation}
p_{mn}  = \! \sum_{k = B - m}^\infty\! \Pr\l( Z_t = k + PI(m \geq P) \r), \quad n = B. \label{Eq:TXProb:b}
\end{equation}}
Combining \eqref{Eq:TXProb:a} and \eqref{Eq:TXProb:b} yields that 
\begin{equation}
p_{mn}  =\l\{ 
\begin{aligned}
& \Pr\l( Z_t = n - m\r), && m < P, n <  B\\
& \Pr\l( Z_t = n - m + P\r), && m \geq P, n <  B\\
& \sum_{k = B - m}^\infty \Pr\l( Z_t = k\r), && m < P, n =  B\\
& \sum_{k = B - m}^\infty \Pr\l( Z_t = k + P\r), && m \geq P, n =  B. 
\end{aligned}
\r. \label{Eq:TXProb}
\end{equation}
Let $\pi_m$ denote the steady-state probability of state $m$ of the Markov chain $\{S_t\}$. Moreover, 
let $\mathbf{P}$ represent the transition-probability matrix with the $(m, n)$-th element given by $p_{mn}$ and $\boldsymbol{\pi}$ the steady-state-probability row vector with the $m$-th element given by $\pi_m$.   Applying Perron-Frobenius theorem \cite{GallagerBook:StochasticProcs:95}, 
\begin{equation}\label{Eq:Markov:Stationary}
\boldsymbol{\pi}\bP = \boldsymbol{\pi}. 
\end{equation}
The stationary probabilities  $\{\pi_m\}$ can be computed by solving  \eqref{Eq:Markov:Stationary} under  the constraint $\sum_{m=0}^{B} \pi_m=1$. 
Given $\boldsymbol{\pi}$, the transmission probability is obtained as  $\rho = \sum_{m = P}^B \pi_m$. 

\begin{figure}[t]
\begin{center}
\includegraphics[width=8.5cm]{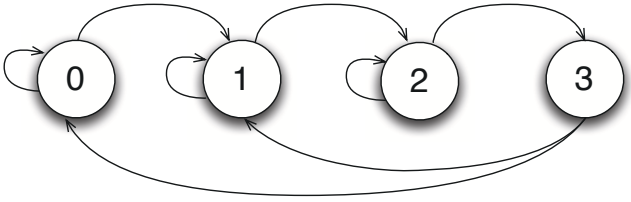}
\caption{A Markov chain modeling the battery level at the typical transmitter for the special case of binary energy arrivals and $P = 3$. }
\label{Fig:Markov}
\end{center}
\end{figure}

The stationary probabilities $\{\pi_m\}$ can be derived in closed form for the simple case of binary energy arrivals, namely  that $Z_t\in \{0, 1\}$.  Then energy-arrival rate is $\lambda_e = \E[Z_t] = \Pr(Z_t = 1)$. The   corresponding transition probabilities are modified  from  \eqref{Eq:TXProb} as 
\begin{equation}\label{Eq:TXProb:Binary}
p_{mn}  =\l\{ 
\begin{aligned}
& 1-\lambda_e,  && (m = P, n = 0) \ \text{or} \  (m <  P,  n = m)\\
& \lambda_e,  && (m = P, n = 1) \ \text{or} \ (m <  P,  n = m+1)\\
& 0, && \text{otherwise}. 
\end{aligned}
\r. \nn
\end{equation}
Based on above  transition  probabilities, the Markov chain $\{S_t\}$ is illustrated in Fig.~\ref{Fig:Markov}. The stationary distribution $\{\pi_m\}$ satisfies the following equations obtained from 
 \eqref{Eq:Markov:Stationary} and the above expression for $p_{mn}$:
 \begin{equation}\begin{aligned}
\pi_0 &= \pi_P(1-\lambda_e)+  \pi_0(1-\lambda_e)  \\
\pi_1 &= \pi_P\lambda_e+  \pi_1(1-\lambda_e) + \pi_0\lambda_e  \\
\pi_m & = \pi_{m-1}\lambda_e +   \pi_m(1-\lambda_e), \qquad m=2, \cdots,  P-1 \\
\pi_P & = \pi_{P-1}\lambda_e. 
\end{aligned}\label{Eq:Markov:Eqs}
\end{equation} 
Solving the equations in \eqref{Eq:Markov:Eqs} and $\sum_{m=0}^P\pi_m  = 1$ gives the following  proposition. 
\begin{proposition}\label{Prop:BinaryArrival}\emph{Consider the case that  energy arrivals are binary ($Z_t \in \{0, 1\}$) and the  transmission power $P$ is a positive integer  no larger than the battery capacity $B$. The distribution of the battery level $S_t$  is 
 \begin{align}
 \pi_0 &= \frac{1-\lambda_e}{P}, \qquad  \pi_{P} = \frac{\lambda_e}{P } \nn\\
\pi_m & = \frac{1}{P}, \qquad m = 1, \cdots, P-1.\nn
\end{align}
The probability for battery overflow is zero  and the transmission probability is $\rho = \pi_P = \lambda_e/P$. 
}
\end{proposition} 

\section{Network Throughput }\label{Section:Thput}
In this section, using results on transmission probability as derived in the preceding section,  the network throughput is maximized by optimizing transmission power under the outage-probability constraint and assuming  infinite battery capacity. It is straightforward though tedious to extend the network-throughput  analysis to the case of finite battery capacity using related  results from the last section.\footnote{Essentially, the network throughput  for the case of finite battery capacity can be bounded by modifying  the current analysis   such that the transmission probability $\rho$ is replaced with its bounds as specified in Proposition~\ref{Prop:TXProb:FiniteBattery}.}   Such an extension provides  few new insights and thus is omitted.

\subsection{Maximum Network Throughput}
To characterize the network throughput, transmission power $P$ is related to the active  transmitter density $\lambda_t$ under the outage-probability constraint. To this end, we define a parameter $\mu_\epsilon$, called the \emph{nominal node density},  as the density of a homogeneous PPP $\Lambda(\mu_\epsilon)$ such that 
\begin{equation}\label{Eq:NomDen}
\Pr\l(\sum_{T\in\Lambda(\mu_\epsilon)} |T|^{-\alpha} > 1\r)=\epsilon 
\end{equation}
where the summation can be interpreted as  the interference power measured at a receiver located at the origin from unit-power interferers distributed as $\Lambda(\mu_\epsilon)$. 
Note that $\mu_\epsilon$ depends only on $\epsilon$, { $\alpha$} and the distribution of a PPP  and is independent of  other network parameters. Moreover, 
$\mu_\epsilon$ is  a strictly-monotone-increasing function of $\epsilon$ due to the fact that denser interferers result in larger outage probability for a link. The expression of $\mu_\epsilon$ has no closed form and its value can  be computed by  simulation (see e.g., \cite{WeberKam:CompComplexMANETs:2006}).  The relation between $\mu_{\epsilon}$ and $\epsilon$ is shown in Fig.~\ref{Fig:NormDen}.  The following lemma results  from Mapping Theorem \cite[p18]{Kingman93:PoissonProc}. 

\begin{lemma}\label{Lem:PP:Transform}\emph{Consider the homogeneous PPP $\Lambda(\mu_\epsilon)=\{X\}$. The process $a \Lambda(\mu_\epsilon)=\{aX\}$ with $a > 0$ is a homogeneous PPP with density $\mu_\epsilon/a^2$. 
}
\end{lemma}

\begin{figure*}
\begin{center}
\includegraphics[width=11cm]{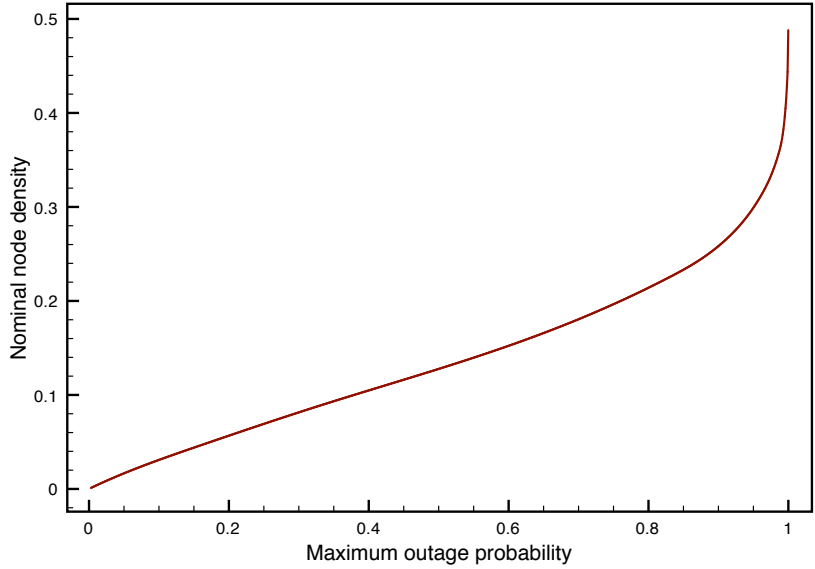}
\caption{The relation between the nominal node density $\mu_{\epsilon}$ and the maximum outage probability $\epsilon$}
\label{Fig:NormDen}
\end{center}
\end{figure*}

Define the admissible  set $\mathcal{F}$ as all combinations of   $(\lambda_t, P)$ that satisfy   the outage-probability  constraint:
\begin{equation}\label{Eq:FeasibleSet:Def}
\mathcal{F} = \{(\lambda_t, P) \in \mathds{R}^+\times \mathds{R}^+\mid \Pout(\lambda_t, P) \leq \epsilon\}. 
\end{equation}
A combination $(\lambda_t, P)$ is \emph{admissible} if it belongs to $\mathcal{F}$. 
{To derive $\mathcal{F}$, since $\Pi$ follows the same distribution as $\Lambda(\lambda_t)$,  $\Pout$ in 
\eqref{Eq:Pout:a} can be rewritten as 
\begin{equation}
\Pout =\Pr\l(\sum_{T\in\Lambda(\lambda_t)}|T|^{-\alpha} > \frac{1}{\theta}-\frac{1}{P} \r).\nn 
\end{equation}
According to  Lemma~\ref{Lem:PP:Transform}, $\Lambda(\lambda_t)$ has the same distribution as $a\Lambda(\mu_{\epsilon})$ with $a = \sqrt{\mu_{\epsilon}/\lambda_t}$. Consequently 
\begin{align}
\Pout &=\Pr\l(\sum_{T\in a \Lambda(\mu_{\epsilon})}|T|^{-\alpha} > \frac{1}{\theta}-\frac{1}{P} \r) \nn\\
& = \Pr\l(\sum_{T\in \Lambda(\mu_{\epsilon})}|aT|^{-\alpha} > \frac{1}{\theta}-\frac{1}{P} \r)\nn\\
& = \Pr\l(\sum_{T\in \Lambda(\mu_{\epsilon})}|T|^{-\alpha} >\l(\frac{\mu_{\epsilon}}{\lambda_t}\r)^{\frac{\alpha}{2}} \l(\frac{1}{\theta}-\frac{1}{P} \r)\r). \label{Eq:Pout:b}
\end{align}
}Combining \eqref{Eq:NomDen}, \eqref{Eq:FeasibleSet:Def} and \eqref{Eq:Pout:b} leads to the following lemma. 

\begin{lemma}\label{Lem:Feasbility}\emph{The admissible  set $\mathcal{F}$ is given as 
\begin{equation}\label{Eq:Feasible:Set}
\mathcal{F} \rightarrow \l\{(\lambda_t, P)\in\mathds{R}^+\times \mathds{R}^+\mid \lambda_t \leq \zeta(P) \r\} 
\end{equation}
where $\zeta(P)$ represents the network interference temperature and is given as
\begin{equation}\label{Eq:IntTemp}
\zeta(P) = \mu_\epsilon \l(\frac{1}{\theta} - \frac{1}{P}\r)^{\frac{2}{\alpha}}, \qquad P \geq \theta. 
\end{equation}
}
\end{lemma}

\begin{remark}\emph{The network interference temperature $\zeta(P)$ specifies the maximum density of interferers a link  can tolerate without violating the outage-probability constraint. This quantity  is analogous to  the \emph{interference temperature} in cognitive-radio  systems that  measures  the maximum amount  of additional  interference  for a certain frequency band without significantly degrading the reliability of  communications therein \cite{Haykin:CognitiveRadio:2005}.
On one hand, the network for large $P$  is interference-limited and further increasing $P$ does not contribute any network-throughput gain. 
Correspondingly, it can be observed from  \eqref{Eq:IntTemp} that $\zeta(P)$ saturates  as $P\rightarrow\infty$:
\begin{equation}\label{Eq:Lim:NIT}
\lim_{P\rightarrow\infty}\zeta(P) = \frac{\mu_{\epsilon}}{\theta^{\frac{2}{\alpha}}}. 
\end{equation}
On the other hand, the network for small $P$ is noise-limited. The value of 
$\zeta(P)$  is not well defined if $P$ is below the target SINR $\theta$, for which the outage constraint cannot be satisfied even in the absence of interference.}
\end{remark}

\begin{remark}\emph{The admissible  set $\mathcal{F}$ is illustrated  as the shaded region in Fig.~\ref{Fig:Feasibility}.  Increasing  $\mu_\epsilon$ (relaxing the outage-probability constraint) enlarges $\mathcal{F}$ and vice versa. Let $\lambda^*$ denote the maximum of the admissible values for $\lambda_t$ under the outage constraint. Then the  boundary of $\mathcal{F}$ corresponds to $\lambda^* = \zeta(P)$.} 
\end{remark}

\begin{remark}\emph{The condition $\lambda_t \leq \zeta(P)$ for the admissible set in \eqref{Eq:Feasible:Set} guarantees that $P$ (equal to the received SNR) is  no smaller   than $\theta$ that is the minimum received SINR required for correct decoding.}
\end{remark}

We are ready to derive the maximum network throughput $R^*$ and the optimal transmission power $P^*$. Let $f$ denote the function that maps $P$ to $\lambda_t$ for fixed $\lambda_0$ and $\lambda_e$, which is obtained from Theorem~\ref{Theo:TXProb:InfBattery} as
\begin{equation}
f(P) = \lambda_0\min\l(1, \frac{\lambda_e}{P}\r). \nn
\end{equation}
The derivation of $R^*$ can be intuitively explained using Fig.~\ref{Fig:Feasibility}.  It can be observed form the  curve depicting the function   $f(P)$ that   $f(P)$ is fixed at $\lambda_0$  for all $P \leq\lambda_e$ and a strictly-monotone-decreasing function for $P > \lambda_e$. Let $P^*$ correspond to the intersection  between the curve $(P, f(P))$ and the admissible  set $\mathcal{F}$.  Then under the outage-probability constraint, $f(P)$ is maximized at $P = P^*$, which corresponds to $R^*$  based on the definition in \eqref{Eq:Thput:Def} since it  is proportional to $\lambda_t$.  Note that the  claim is based on the assumption $\epsilon \ll 1$ or else need not hold (see Remark~\ref{Rem:LargePout}). Depending on if $P^*$ is larger or smaller than $\lambda_e$, $P^*$ and $R^*$ have different expressions as shown in Theorem~\ref{Theo:TPut:Long:NonEMR} that states   the main result of this section. 

\begin{figure*}
\begin{center}
\includegraphics[width=14cm]{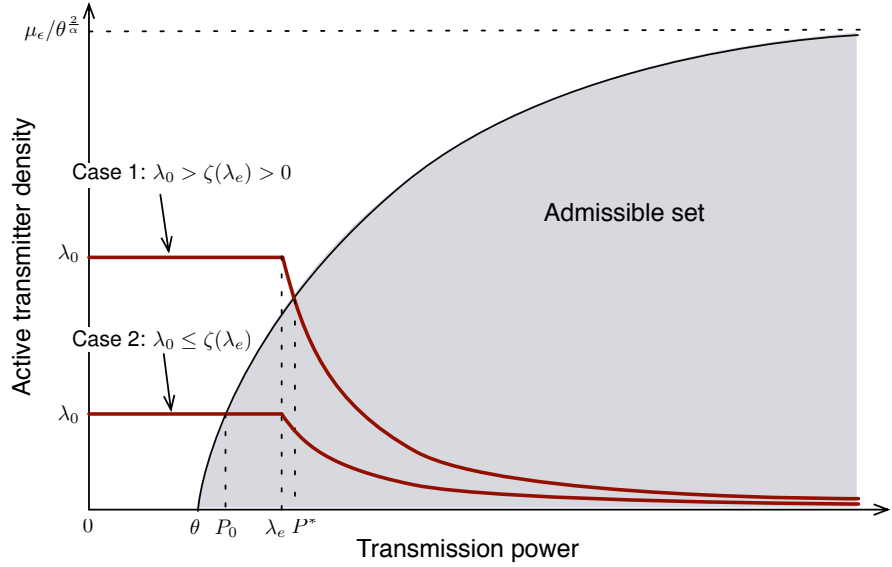}
\caption{The admissible   set $\mathcal{F}$ is sketched  as the shaded region that contains  all combinations of active  transmitter density $\lambda_t$ and transmission power $P$ that satisfy  the outage-probability constraint. Given infinite battery capacity and fixed transmitter density $\lambda_0$,  $f(P)$ is a monotone-decreasing function of $P$ as plotted with the thick lines based on Theorem~\ref{Theo:TXProb:InfBattery} for the cases of $\lambda_0 > \zeta(\lambda_e) > 0$ and $\lambda_0 \leq \zeta(\lambda_e)$, where the intersections are indicated by $P^*$ and $P_0$, respectively. Note that $P^*$ and an arbitrary value in $[P_0, \lambda_e]$ give the optimal  transmission power that maximizes the network throughput for the corresponding  cases.}
\label{Fig:Feasibility}
\end{center}
\end{figure*}

\begin{theorem} \label{Theo:TPut:Long:NonEMR}\emph{Given  infinite battery capacity, the maximum network throughput $R^*$ and  the optimal transmission power $P^*$ are specified  as follows. 
\begin{enumerate}
\item If $\lambda_0 \leq \zeta(\lambda_e)$,  
\begin{equation}\label{Eq:Max:TPut:a}
R^* = \lambda_0\log_2(1+ \theta)
\end{equation}
and $P^*$ is  an arbitrary value in the range $\l[\frac{\theta}{1 -  \theta (\lambda_0/\mu_\epsilon)^{\frac{\alpha}{2}}}, \lambda_e\r]$. 
\item If $\lambda_0 > \zeta(\lambda_e) > 0$,
\begin{equation}\label{Eq:Max:TPut:b}
R^* = \frac{\lambda_0  \lambda_e}{P^*}\log_2(1+\theta)
\end{equation}
where $P^* \geq \lambda_e$ and $\sqrt{P^*}$ solves the following polynomial  equation:
\begin{equation}\label{Eq:Poly:Func}
x^{\alpha} - \theta x^{\alpha -2} - \theta \l(\frac{\lambda_0\lambda_e}{\mu_{\epsilon}}\r)^{\frac{\alpha}{2}}=0. 
\end{equation}
\end{enumerate}}
\end{theorem}
\begin{proof}  First, consider the case of $\lambda_0 \leq \zeta(\lambda_e)$.   The throughput expression  in \eqref{Eq:Thput:Def:a} implies that
\begin{equation}\label{Eq:MaxThput}
R^* \leq \lambda_0 \log(1+\theta). 
\end{equation}
Define $P_0$ such that $\lambda_0 = \zeta(P_0)$ (see Fig.~\ref{Fig:Feasibility}).   Using  the definition of $\zeta$ in \eqref{Eq:IntTemp},  
\[
P_0 = \frac{\theta}{1 -  \theta (\lambda_0/\mu_\epsilon)^{\frac{\alpha}{2}}}. 
\]
It can be  observed from \eqref{Eq:IntTemp} that  $\zeta$ is a strictly-monotone-increasing function. As a result, since $\lambda_0 \leq \zeta(\lambda_e)$ and   $\lambda_0 = \zeta(P_0)$,   $P_0 \leq \lambda_e$  and thus the set $[P_0, \lambda_e]$ (the range of $P^*$ in the theorem statement) is nonempty.  Consider an arbitrary value $p \in [P_0, \lambda_e]$. Since    $\lambda_0 \leq \zeta(p)$ from the monotonicity of  $\xi$,   $(\lambda_t, P) = (\lambda_0, p)$ is admissible  according to  Lemma~\ref{Lem:Feasbility}. Furthermore, $(\lambda_t, P) = (\lambda_0, p)$ is feasible as $\rho(p) = 1$ based on Theorem~\ref{Theo:TXProb:InfBattery} and $p \leq \lambda_e$. It follows that $P=p$  maximizes $R$ by achieving the equality in \eqref{Eq:MaxThput}. This proves the desired result for the case of $\lambda_0 \leq \zeta(\lambda_e)$. 

Next, consider the other case of $\lambda_0 > \zeta(\lambda_e)>0$. Using the throughput expression in \eqref{Eq:Thput:Def:a} and Lemma~\ref{Lem:Feasbility}, the problem of  maximizing  the network throughput is equivalent to 
\begin{equation}\label{Eq:Thput:MaxProb}
\begin{aligned}
&\text{maximize}&&  \rho(P)\\
 &\text{subject to}&& \lambda_0\rho(P) \leq \zeta(P). 
 \end{aligned}
\end{equation}
The inequality $P^* \geq \lambda_e$ in the theorem statement can be proved by contradiction as follows. 
Assume that $P^* < \lambda_e$. This assumption results in  $\rho(P^*) = 1$ by applying Theorem~\ref{Theo:TXProb:InfBattery}. Moreover, $\zeta(P^*) < \zeta(\lambda_e)$ is obtained using  the aforementioned monotonicity of  $\zeta$,   and hence $\lambda_0 > \zeta(P^*)$ given that  $\lambda_0 > \zeta(\lambda_e)$. Combining $\lambda_0 > \zeta(P^*) $ and $\rho(P^*) = 1$ shows that  the assumption of $P^* < \lambda_e$ violates  the  constraint in \eqref{Eq:Thput:MaxProb}, proving that  $P^* \geq \lambda_e$. Then applying Theorem~\ref{Theo:TXProb:InfBattery} yields the desired  result in \eqref{Eq:Max:TPut:b}. Last, since $\rho(P)$ and $\zeta(P)$ are strictly-monotone-decreasing and strictly-monotone-increasing functions, respectively, the solution $P^*$ for the  problem in \eqref{Eq:Thput:MaxProb} must satisfy  $\lambda_0\rho(P^*) = \zeta(P^*)$ or equivalently $\sqrt{P^*}$ solves  the polynomial equation in \eqref{Eq:Poly:Func}, completing the proof. 
\end{proof}

\begin{remark}\emph{For the case of $\lambda_0 \leq \zeta(\lambda_e)$, the network is relatively sparse  and $\lambda_e$ is sufficiently large such that it is optimal as well as feasible for all transmitters to transmit  with probability one, resulting in the network throughput in \eqref{Eq:Max:TPut:a}. For the case of $\lambda_0 > \zeta(\lambda_e)$, the network is relatively dense and  high transmission power is required for satisfying the outage-probability constraint. Consequently, 
not all transmitters can transmit simultaneously,  corresponding to  the optimal transmission probability  smaller than one and the network throughput in \eqref{Eq:Max:TPut:b}.}
\end{remark}

\begin{remark}\emph{With $\alpha > 2$ and the last coefficient at the right-hand  side being negative, the polynomial equation in \eqref{Eq:Poly:Func} has at least one strictly positive solution that gives  $P^*$  for the case of $\lambda_0 > \zeta(\lambda_e)$. For the special case of $\alpha = 4$, the  polynomial equation in \eqref{Eq:Poly:Func} is quadratic and solving it gives $P^*$ in closed form as shown below.}
\end{remark}

\begin{corollary} \emph{Given infinite battery capacity,  $\lambda_0 > \zeta(\lambda_e)$ and $\alpha = 4$, the optimal transmission power $P^*$ is 
\begin{equation}
P^* = \frac{\theta + \sqrt{\theta^2 + 4\theta\l(\frac{\lambda_0\lambda_e}{\mu_{\epsilon}}\r)^{\frac{\alpha}{2}}}}{2}. \nn
\end{equation}
}
\end{corollary}

\begin{remark}\label{Rem:LargePout}\emph{Recall that the throughput maximized in Theorem~\ref{Theo:TPut:Long:NonEMR} is defined in \eqref{Eq:Thput:Def} based on the assumption $\epsilon \ll 1$, where the scaling factor $(1-\Pout)$ (success probability)  is omitted for simplicity since it is close to one under the outage-probability constraint. If this factor is considered, 
changing the value of $P^*$ over the range $\l[\frac{\theta}{1 -  \theta (\lambda_0/\mu_\epsilon)^{\frac{\alpha}{2}}}, \lambda_e\r]$ [see Theorem~\ref{Theo:TPut:Long:NonEMR} for the case of $\lambda_0 \leq \zeta(\lambda_e)$] can lead to a throughput variation  no larger than $\epsilon R^*$, which is negligible given $\epsilon \ll 1$.  However,  if $\epsilon$ is comparable with one or there is no outage constraint ($\epsilon =1$),  the success probability  should be accounted for and the throughput redefined as
\begin{equation}\nn
R= (1-\Pout)\rho \lambda_0   \log_2(1+\theta) 
\end{equation}
where $\Pout \leq \epsilon$. The results in Theorem~\ref{Theo:TPut:Long:NonEMR} can be extended using the redefined  metric by analyzing  $\Pout$ as a function of $P$, which has no closed-form but can be approximated by its bounds \cite{WeberAndrews:TransCapWlssAdHocNetwkSIC:2005}. 
}
\end{remark}

\subsection{Maximum Network Throughput: Extreme Cases}
Consider a network with  a high energy-arrival rate ($\lambda_e\rightarrow\infty$). 
The maximum    network throughput  $R^*$ can be upper bounded as 
\begin{equation}\label{Eq:Thput:UB}
R^* \leq \frac{\mu_{\epsilon}}{\theta^{\frac{2}{\alpha}}}\log_2(1+\theta)
\end{equation}
since the outage-probability constraint requires that (see Lemma~\ref{Lem:Feasbility})
\begin{align}
\lambda_t &\leq \zeta(P)\nn\\
&\leq \frac{\mu_{\epsilon}}{\theta^{\frac{2}{\alpha}}} \label{Eq:NIT:UB}
\end{align}
where \eqref{Eq:NIT:UB} follows from  \eqref{Eq:Lim:NIT} and that $\zeta$ is a monotone-increasing function. 
Combining the two upper bounds on $R^*$ in \eqref{Eq:MaxThput} and \eqref{Eq:Thput:UB} gives 
\begin{equation}\label{Eq:Thput:UB:Combine}
R^* \leq \min\l(\lambda_0, \frac{\mu_{\epsilon}}{\theta^{\frac{2}{\alpha}}}\r)  \log_2(1+\theta). 
\end{equation}
For a high energy-arrival rate, equality is achieved in \eqref{Eq:Thput:UB:Combine} as shown below. 
\begin{proposition}\label{Prop:HighEnergy}\emph{Given infinite battery capacity, as the energy-arrival rate $\lambda_e\rightarrow\infty$, the maximum network throughput  converges as 
\begin{equation}\label{Eq:Thput:HighEnergy}
\lim_{\lambda_e\rightarrow\infty}R^*(\lambda_e) = 
\min\l(\lambda_0, \frac{\mu_{\epsilon}}{\theta^{\frac{2}{\alpha}}}\r)  \log_2(1+\theta). 
\end{equation}
}
\end{proposition}

\begin{proof} 
First, consider  the case of $\lambda_0 \leq \mu_{\epsilon}\theta^{-\frac{2}{\alpha}}$. Set   $P = \lambda_e - \delta$ with $\delta > 0$. This results in $\rho =1$ according to Theorem~\ref{Theo:TXProb:InfBattery}. Consequently, $\lambda_0 = \lambda_t$ and hence  $\lambda_t \leq \mu_{\epsilon}\theta^{-\frac{2}{\alpha}}$  from the assumption about $\lambda_0$. Combining this inequality and  \eqref{Eq:Lim:NIT} yields that  $\lambda_t \leq \zeta(P)$ as $P \rightarrow\infty$ along with $\lambda_e\rightarrow\infty$. It follows that as $\lambda_e\rightarrow\infty$, the combination $(\lambda_t, P) = (\lambda_0, \lambda_e-\delta)$ is admissible according to Lemma~\ref{Lem:Feasbility}. This proves  the equality in \eqref{Eq:Thput:UB:Combine} for the current case. 
 
Next, consider the case of  $\lambda_0 >  \mu_{\epsilon}\theta^{-\frac{2}{\alpha}}$. Given this strict inequality,  there exists $\delta > 0$ such that 
\begin{equation}\label{Eq:Density:LB}
\lambda_0 >  \frac{\mu_{\epsilon}}{\theta^{\frac{2}{\alpha}}}-\delta.
\end{equation}  
Set $P$ as 
\begin{equation}\label{Eq:P:Choice}
P =\frac{\lambda_e\lambda_0}{\mu_{\epsilon}\theta^{-\frac{2}{\alpha}} - \delta}. 
\end{equation}
Combining \eqref{Eq:Density:LB} and \eqref{Eq:P:Choice} gives  $ \lambda_e < P$. Consequently, applying Theorem~\ref{Theo:TXProb:InfBattery} gives 
\[
\rho = \frac{\mu_{\epsilon}\theta^{-\frac{2}{\alpha}} - \delta}{\lambda_0}
\]
and hence 
\begin{equation}
\lambda_t =\frac{\mu_{\epsilon}}{\theta^{\frac{2}{\alpha}}} - \delta.  \label{Eq:Den:Exp}
\end{equation}
 As $P\rightarrow\infty$ along with $\lambda_e\rightarrow\infty$, it follows from \eqref{Eq:Lim:NIT} and \eqref{Eq:Den:Exp} that 
there exists $\delta > 0$ such that  $\lambda_t \leq \zeta(P)$. As a result, by applying Lemma~\ref{Lem:Feasbility},  the combination  of $(\lambda_t, P)$ as specified   in  \eqref{Eq:P:Choice}  and\eqref{Eq:Den:Exp} is admissible.  This leads to  
\begin{equation}
\lim_{\lambda_e\rightarrow\infty}R^*(\lambda_e) = \l(\frac{\mu_{\epsilon}}{\theta^{\frac{2}{\alpha}}} - \delta\r)\log_2(1+\theta). \nn
\end{equation}
Letting $\delta \rightarrow 0$ proves   the equality in \eqref{Eq:Thput:HighEnergy} for the current case, completing the proof. 
\end{proof}

\begin{remark}\emph{Given a high energy arrival rate and infinite battery capacity, in the steady state, transmitters always have sufficient energy for transmission. Therefore, the expression in \eqref{Eq:Thput:HighEnergy} also specifies   the maximum network throughput of a MANET with reliable power supplies instead of energy harvesting.}
\end{remark}

\begin{remark}\emph{For the case of $\lambda_0 < \mu_{\epsilon}\theta^{-\frac{2}{\alpha}}$,  the active transmitter density  is below the network-interference temperature even though all transmitters transmit with probability one. Therefore, there is  margin for further increasing active transmitter density   without violating the outage-probability constraint. For this reason, the network-throughput limit  in \eqref{Eq:Thput:HighEnergy} for the current case is proportional to the transmitter density. However, for the case of $\lambda_0 \geq \mu_{\epsilon}/\theta^{\frac{2}{\alpha}}$, active transmitter density reaches the network-interference temperature and cannot be further increased. Consequently, the corresponding network-throughput limit  in \eqref{Eq:Thput:HighEnergy} is independent of  the transmitter density.}
\end{remark} 

\begin{figure*}
\begin{center}
\includegraphics[width=11cm]{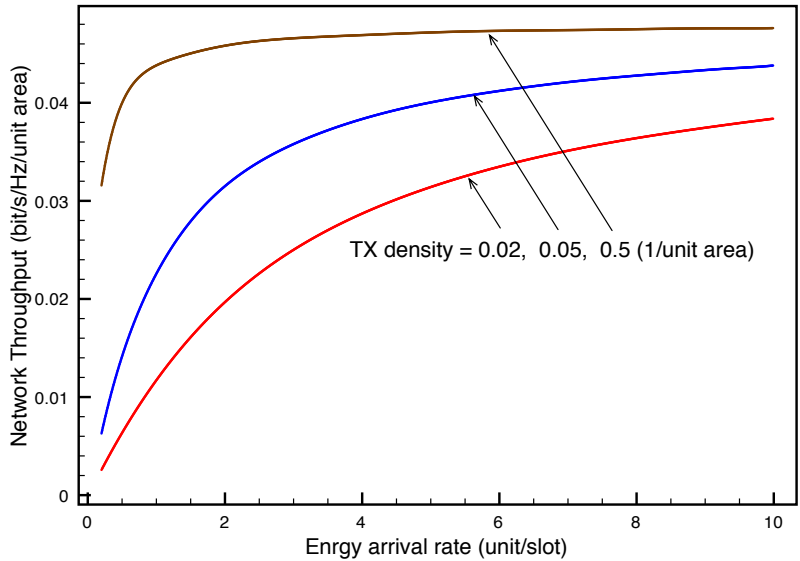}
\caption{Maximum network throughput versus  energy-arrival rate for  optimal  transmission power,  infinite battery capacity, and   the transmitter density  $\lambda_0 = \{0.02, 0.05, 0.5\}$. }
\label{Fig:MaxCap:Var:Energy}
\end{center}
\end{figure*}

Consider  a sparse network ($\lambda_0 \rightarrow 0$). It is optimal for each transmitter to transmit with probability one by setting the transmission power $P^* \in (\theta,  \lambda_e]$ if $\lambda_e > \theta$ or otherwise with probability $\lambda_e/\theta$ by setting $P^*$  to be equal to  $\theta$ (see Theorem~\ref{Theo:TXProb:InfBattery}). The corresponding network throughputs are $R^* = \lambda_0\log_2(1+\theta)$ and 
$R^* = \frac{\lambda_0\lambda_e}{\theta}\log_2(1+\theta)$, respectively, which both diminish as $\lambda_0 \rightarrow 0$.  Next, for a dense network ($\lambda_0 \rightarrow \infty$), the maximum network throughput is specified as follow. 
\begin{proposition}\label{Prop:DenseNet}\emph{Given infinite battery capacity, as the transmitter density   $\lambda_0 \rightarrow\infty$,  the maximum network throughput converges as  
\begin{equation}\label{Eq:Thput:DenseNet}
\lim_{\lambda_0\rightarrow\infty}R^*(\lambda_0) = \frac{\mu_{\epsilon}}{\theta^{\frac{2}{\alpha}}}\log_2(1+\theta). 
\end{equation}
}
\end{proposition}
\begin{proof} Set $P$ as the following function of $\lambda_0$: 
\begin{equation}\label{Eq:P:DenseNet}
P(\lambda_0) =\frac{\lambda_0\lambda_e}{\mu_{\epsilon}}\l(\frac{1}{\theta} - \frac{1}{\log \lambda_0}\r)^{-\frac{2}{\alpha}} 
\end{equation}
that is shown shortly to achieve the limit of $R^*$ in \eqref{Eq:Thput:DenseNet}. Given \eqref{Eq:P:DenseNet},  there exists $\tau_1 > 0$ such that $P(\lambda_0) \geq\lambda_e$ for all  $\lambda_0 > \tau_1$. Therefore,  it follows from \eqref{Eq:Thput:Def:a} and Theorem~\ref{Theo:TXProb:InfBattery} that 
\begin{align}
\lim_{\lambda_0\rightarrow\infty}R(\lambda_0, P(\lambda_0)) & = \lim_{\lambda_0\rightarrow\infty}\frac{\lambda_0\lambda_e}{P(\lambda_0)}\log_2(1+\theta)\nn\\
&= \lim_{\lambda_0\rightarrow\infty}\mu_{\epsilon}\l(\frac{1}{\theta} - \frac{1}{\log \lambda_0}\r)^{\frac{2}{\alpha}}\log_2(1+\theta)\nn\\
&= \frac{\mu_{\epsilon}}{\theta^{\frac{2}{\alpha}}}\log_2(1+\theta).\label{Eq:Thput:DenseNet:a}
\end{align}
Combining \eqref{Eq:Thput:UB} and \eqref{Eq:Thput:DenseNet:a}   shows that the maximum  network throughput has  the limit in \eqref{Eq:Thput:DenseNet} as $\lambda_0$ increases.

The remaining proof verifies that $P(\lambda_0)$ and the corresponding $\lambda_t$ are admissible as $\lambda_0 \rightarrow\infty$. It follows from \eqref{Eq:P:DenseNet} and Theorem~\ref{Theo:TXProb:InfBattery} that for all $\lambda_0 > \tau_1$, $\lambda_t$ is a function of $\lambda_0$ and given as 
\begin{equation}
\lambda_t(\lambda_0) = \mu_{\epsilon}\l(\frac{1}{\theta} - \frac{1}{\log\lambda_0}\r)^{\frac{2}{\alpha}}.  \label{Eq:Density:DenseNet}
\end{equation}
Substituting  \eqref{Eq:P:DenseNet} into \eqref{Eq:IntTemp} yields 
\begin{equation}\label{Eq:NT:Dense}
\zeta(P(\lambda_0)) =  \mu_{\epsilon}\l(\frac{1}{\theta} - \frac{\mu_{\epsilon}}{\lambda_0\lambda_e}\l(\frac{1}{\theta} - \frac{1}{\log\lambda_0}\r)^{\frac{2}{\alpha}}\r)^{\frac{2}{\alpha}}. 
\end{equation}
By comparing \eqref{Eq:Density:DenseNet} and \eqref{Eq:NT:Dense}, there exists $\tau_2 > 0$ such that $\lambda_t(\lambda_0) \leq \zeta(P(\lambda_0))$ for all $\lambda_0 \geq \tau_2$. This proves the admissibility of $P(\lambda_0)$ in \eqref{Eq:P:DenseNet} and $\lambda_t(\lambda_0)$ in \eqref{Eq:Density:DenseNet} as $\lambda_0 \rightarrow\infty$, completing the proof. 
\end{proof}

\begin{remark}\emph{The rate of total energy  harvested per unit area is $\lambda_t P^* = \lambda_0 \lambda_e$ as $\lambda_0\rightarrow\infty$. The linear growth of the rate  with increasing $\lambda_0$ is due to that the  harvester density is equal to  $\lambda_0$. However, more aggressive energy harvesting in a dense network does not continuously increase the network throughput that saturates at high transmitter power  as the network becomes  interference limited (see  Proposition~\ref{Prop:DenseNet}). This issue may be resolved by using an alternative multiple-access protocol such as frequency-hopping multiple access that reduces the density of simultaneous co-channel transmitters.} 
 \end{remark}

\section{Numerical  Results}\label{Section:Simulation}

The nominal node density $\mu_\epsilon$ is fixed as $0.05$ for all  numerical  results, corresponding to the maximum outage probability $\epsilon \approx 0.015$. The relation between $\mu_\epsilon$ and $\epsilon$ is obtained by simulation based on the following procedure (see e.g.,   \cite{WeberKam:CompComplexMANETs:2006}). The summation over the PPP $\Lambda(\mu_\epsilon)$ in \eqref{Eq:NomDen} is approximated by the signal power measured at the origin due to unit-power transmissions by transmitters uniformly distributed in a disk. The number of transmitters follows the Poisson distribution with mean $200$ and  the disk radius is adjusted such that the expected transmitter density is equal to $\mu_\epsilon$.  Based on this setup, the values of  $(\epsilon, \mu_\epsilon)$ are computed using the Monte Carlo method that yields the plot in Fig.~\ref{Fig:NormDen}.  In addition, all numerical  results are based on the SINR threshold $\theta = 3$ and the path-loss exponent $\alpha = 3$.
\begin{figure*}
\begin{center}
\includegraphics[width=11cm]{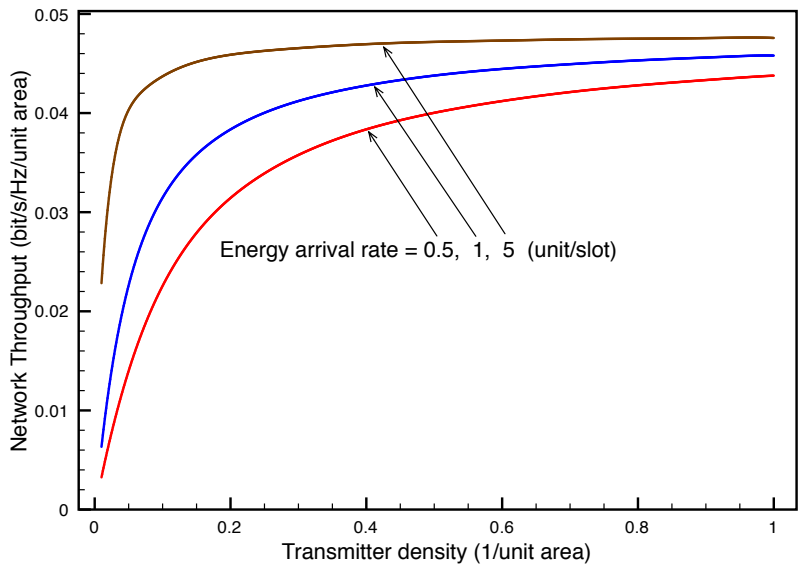}
\caption{Maximum network throughput versus  transmitter density  for  optimal  transmission power,  infinite battery capacity, and the energy-arrival rate $\lambda_e = \{0.5, 1, 5\}$. }
\label{Fig:MaxCap:Vary:Den}
\end{center}
\end{figure*}

\begin{figure*}
\begin{center}
\includegraphics[width=11cm]{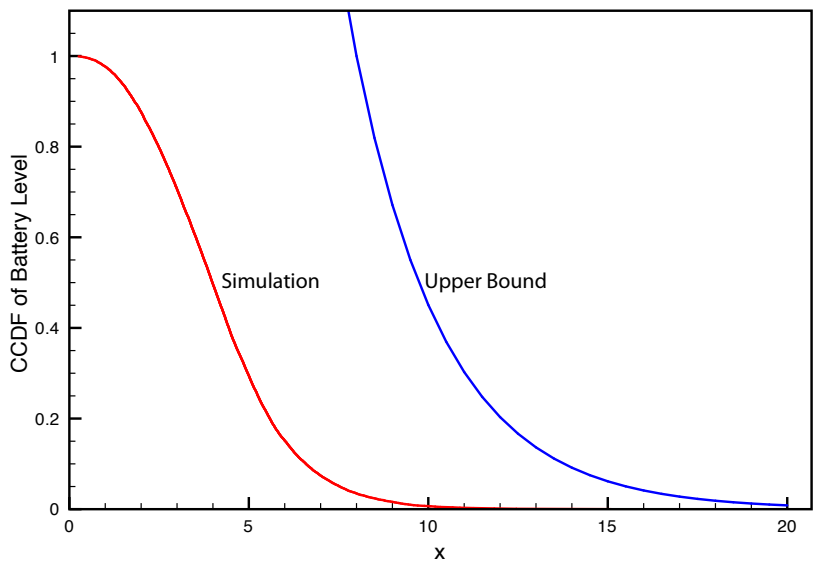}
\caption{A comparison between the average tail probability  of the battery level, $\lim_{n\rightarrow\infty}\frac{1}{n}\sum_{t=1}^n\Pr(S_t> x)$, evaluated  by simulation and its upper bound computed based on Lemma~\ref{Lem:St:Cross} for infinite battery capacity, the DoF of the energy-arrival process $d = 4$, the energy-arrival rate $\lambda_e = 2$ and the transmission power $P = 4$. }
\label{Fig:Energy:CCDF}
\end{center}
\end{figure*}

The distribution of the energy arrival process $\{Z_t\}$ is specified   as follows. 
Let  $\{V_t\}$ denote   an i.i.d. sequence of random variables following the chi-squared distribution with $d\in \{1, 2, \cdots\}$ degrees of freedom (DoF) and mean equal to $d\lambda_e$. Let $\{Z_t\} =  \{\frac{1}{d}V_t\}$ and hence $Z_t$ has mean $\lambda_e$ and variance $2\lambda_e^2/d$. The chosen distribution of $Z_t$ allows its variance (randomness) to be controlled by varying $d$ while the mean of $Z_t$ is fixed. Note that $Z_t$ converges to a constant $\lambda_e$ in probability as $d\rightarrow\infty$ by the law of large numbers.

Infinite battery capacity is assumed for the numerical  results presented in Fig.~\ref{Fig:MaxCap:Var:Energy} and Fig.~\ref{Fig:MaxCap:Vary:Den}. In Fig.~\ref{Fig:MaxCap:Var:Energy}, the  maximum network throughput $R^*$ computed using Theorem~\ref{Theo:TPut:Long:NonEMR}  is plotted against the increasing energy-arrival rate $\lambda_e$  for the transmitter density $\lambda_0 = \{0.02, 0.05, 0.5\}$. It can be observed from Fig.~\ref{Fig:MaxCap:Var:Energy}  that $R^*$ grows as $\lambda_e$ increases and saturates for large $\lambda_e$.  The limits agree with those computed using
Proposition~\ref{Prop:HighEnergy}, namely  $0.04$ bit/s/Hz/unit-area  for 
$\lambda_0 = 0.02$ and $0.048$ bit/s/Hz/unit-area for 
$\lambda_0 = \{0.05, 0.5\}$. In addition, Fig.~\ref{Fig:MaxCap:Var:Energy} shows that in a denser network   (i.e., $\lambda_e = 0.5$), $R^*$  reaches its limit more rapidly as $\lambda_e$ increases.

Fig.~\ref{Fig:MaxCap:Vary:Den} shows the curves of $R^*$ versus  $\lambda_0$ for $\lambda_e = \{0.5, 1, 5\}$, which are obtained using Theorem~\ref{Theo:TPut:Long:NonEMR}.  As $\lambda_0$ increases and regardless of the value of $\lambda_e$,  $R^*$   is observed  to converge to the  limit $0.048$ bit/s/Hz/unit-area predicted by   Proposition~\ref{Prop:DenseNet}.  Moreover, it is observed from Fig.~\ref{Fig:MaxCap:Vary:Den} that  larger $\lambda_e$ results in faster convergence of $R^*$ to its limit as $\lambda_0$ increases.

The average tail probability of the battery level,  $\lim_{n\rightarrow\infty}\frac{1}{n}\sum_{t=1}^n\Pr(S_t> x)$,  is evaluated by simulation and compared in Fig.~\ref{Fig:Energy:CCDF} with its upper bound from Lemma~\ref{Lem:St:Cross} given infinite battery capacity, $d = 4$, $\lambda_e = 2$ and $P = 4$. The bound is observed to be loose  but sufficient for the analysis. The similar observation and remark also hold for the upper bound on the energy-overshoot function $D_t(x)$ as given in Lemma~\ref{Lem:ExcessEnergy} and the numerical results are omitted for brevity. 

\begin{figure*}
\begin{center}
\includegraphics[width=11cm]{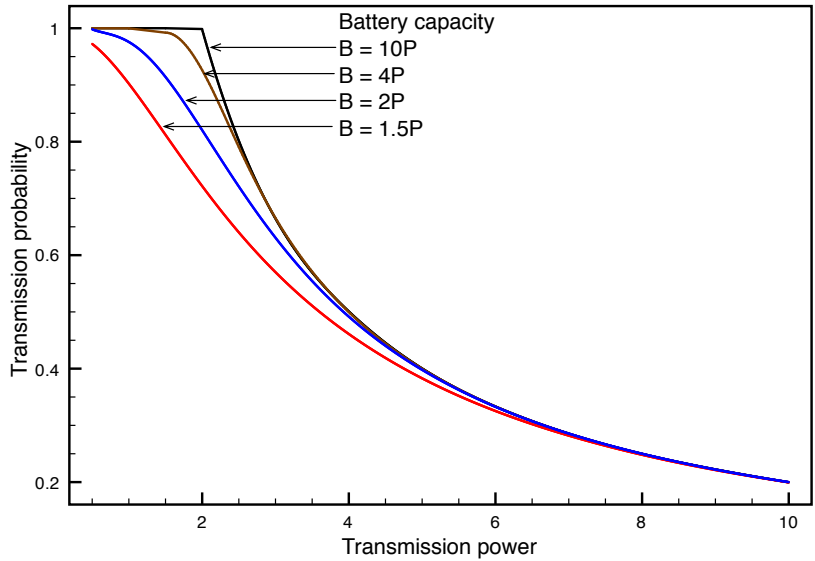}
\caption{Transmission probability versus transmission power for  finite battery capacity $B = \{1.5P, 2P,  4P, 10P\}$,  the DoF of the energy-arrival process $d = 4$, the energy-arrival rate   $\lambda_e = 2$, and the transmitter density $\lambda_0 = 0.02$. }
\label{Fig:TXProb}
\end{center}
\end{figure*}
\begin{figure*}
\begin{center}
\includegraphics[width=11cm]{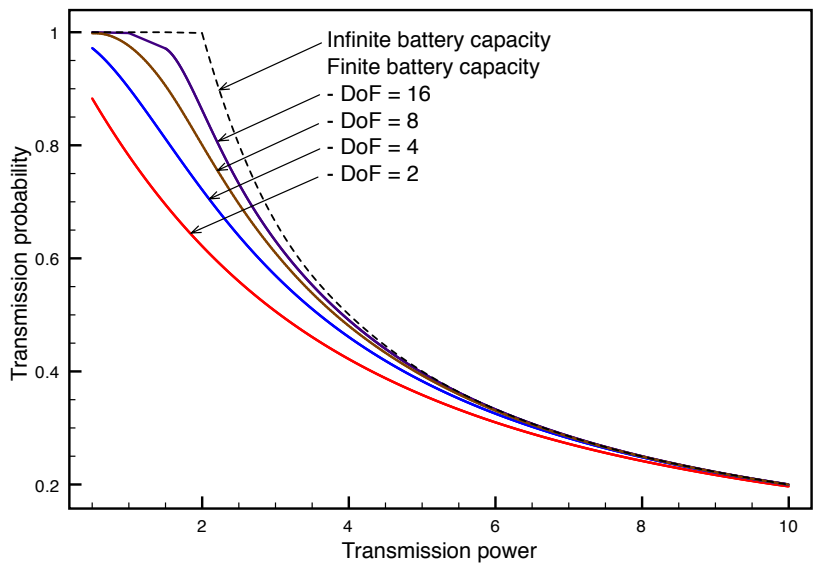}
\caption{Transmission probability versus transmission power for both the cases of finite  ($B = 1.5P$) and infinite battery capacity. The DoF of the energy-arrival process is $d = \{2, 4, 8, 16\}$, the energy-arrival rate $\lambda_e = 2$, and the transmitter density $\lambda_0 = 0.02$. }
\label{Fig:TXProb:VarDoF}
\end{center}
\end{figure*}

Next, consider the case of finite battery capacity. 
In Fig.~\ref{Fig:TXProb}, the transmission probability $\rho$ obtained by simulation is plotted against increasing transmission power $P$ for    finite battery capacity  $B = \{1.5P, 2P, 4P, 10P\}$, $d = 4$, $\lambda_e = 2$, and  $\lambda_0 = 0.02$.   It is found that $B=10P$ is sufficiently large such that the values of $\rho$ closely match those for the case of infinite battery capacity as  computed using Theorem~\ref{Theo:TXProb:InfBattery}. As observed from   Fig.~\ref{Fig:TXProb},   finite battery capacity degrades $\rho$ significantly   only when $P$ and hence $B$   are  relatively small; as $P$ and $B$ increase, $\rho$ rapidly approaches the counterpart for the case of  infinite battery capacity (or that for $B=10P$).

Fig.~\ref{Fig:TXProb:VarDoF} displays the curves of $\rho$ versus   $P$  obtained by simulation for  $B=1.5P$, $d = \{2, 4, 8, 16\}$,  $\lambda_e = 2$ and  $\lambda_0 = 0.02$.   For comparison, the curve for the case of infinite battery capacity is also plotted. As observed from Fig.~\ref{Fig:TXProb:VarDoF}, reducing the randomness of the energy arrival process by increasing $d$  leads to smaller battery-overflow probability and hence higher  $\rho$.  The effect of $d$ on $\rho$ diminishes as $P$  (and hence battery capacity) increases  and $\rho$ converges to its counterpart for the case of infinite battery capacity.

\begin{figure*}
\begin{center}
\includegraphics[width=11cm]{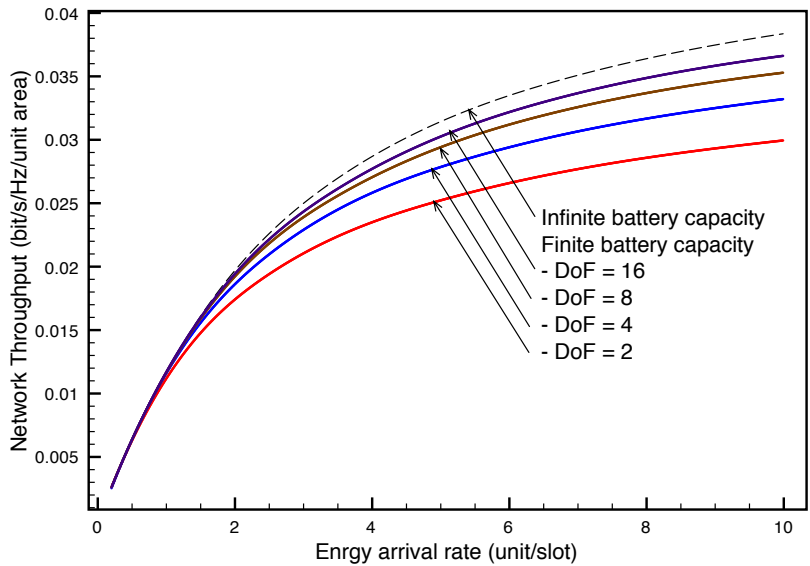}
\caption{Maximum network throughput versus  energy-arrival rate for optimal  transmission power and both the cases of finite  ($B = 1.5P$) and infinite battery capacity. The DoF of the energy-arrival process is $d = \{2, 4, 8, 16\}$ and the transmitter density $\lambda_0 = 0.02$.
}
\label{Fig:Thput:VarDoF}
\end{center}
\end{figure*}

Last, we investigate the effect of the DoF  of the energy arrival process on the network throughput. To this end, Fig.~\ref{Fig:Thput:VarDoF} shows the curves of  $R^*$ versus $\lambda_e$ obtained by simulation for $B = 1.5P$, $d = \{2, 4, 8, 16\}$,  and $\lambda_0 = 0.02$. The curve for the case of infinite battery capacity is also plotted for comparison. It can be observed from Fig.~\ref{Fig:Thput:VarDoF} that  finite battery capacity causes significant  throughput loss especially for large  $\lambda_e$. Such loss is smaller for larger $d$ because of less  randomness in harvested energy and hence smaller battery-overflow probability.

\section{Conclusion}\label{Section:Conclusion}
The energy dynamics in a mobile ad hoc network have  been characterized in terms of transmission probability. Assuming infinite battery capacity, it has been found that the transmission probability is equal to one when the energy-arrival rate exceeds transmission power or otherwise is equal to their ratio. Moreover, for the case of finite battery capacity, bounds on the transmission probability have been obtained and exact expressions have been derived for the special cases of bounded or discrete energy arrivals. 
The results on transmission probability have been applied to derive  the maximum network spatial throughput for a given energy-arrival rate and optimized transmission power.  It has been shown  that it is optimal and feasible for all transmitters to transmit with probability one if the transmitter density is below a threshold that depends on the energy-arrival rate; otherwise, each transmitter should transmit with probability smaller than one.

There are several potential  directions for extending this work. Coexisting wireless networks may harvest electromagnetic (EMR) energy from each others' transmissions. Modeling and designing coexisting networks with EMR energy harvesting  give rise to many new research issues ranging from algorithm design to throughput analysis. The current work focuses on ad hoc networks with random access and can be extended to other types of networks such as cellular networks or other medium-access-control  protocols  such as carrier-sensing multiple access. Last, it is interesting  to investigate the effects of bursty data arrivals and more sophisticated  power control on the throughput of wireless networks powered by energy harvesting.

\section*{Acknowledgement}
The author thanks Rui Zhang for helpful discussion that motivated this research, and the anonymous reviewers whose comments have significantly improved the quality of this paper.

\appendices

\section{Proof of Lemma~\ref{Lem:St:UB}}\label{App:St:UB}
If $t \in \mathcal{T}$, it follows from \eqref{Eq:RandProc2:Def} that $S_t = G'_t$ and hence the inequality in the lemma statement holds since $G_t \geq 0$. Next,  consider the case of $t\notin \mathcal{T}$ and let $t_0\in \mathcal{T}$ denote the time instant  closest to but smaller than $t$.  {  It follows that  the random walk $\{S_t\}$ does not cross the threshold $P$ from below in the time slots $\{t_0+1, \cdots, t\}$. Therefore, if $S_t \geq S_{t_0}$, $S_t  = S_{t_0} + \sum\nolimits_{m=t_0+1}^t \bar{Z}_m$.} Then $S_t$ can be upper bounded as 
\begin{equation}
\begin{aligned}
S_t & \leq S_{t_0} + \max\Big(0, \bar{Z}_t,  \bar{Z}_t + \bar{Z}_{t-1}, \cdots, \\
&\qquad \qquad \sum\nolimits_{m=t_0+1}^t \bar{Z}_m\Big). 
\end{aligned}\label{Eq:St:UB}
\end{equation}
It can be obtained from \eqref{Eq:RandProc1:Def} that 
\begin{equation}
\begin{aligned}
G_t = \max\Big(0, &\bar{Z}_t,  \bar{Z}_t + \bar{Z}_{t-1}, \cdots, \sum\nolimits_{m=t_0+2}^t \bar{Z}_m, \\
&  G_{t_0} +\sum\nolimits_{m=t_0+1}^t \bar{Z}_m\Big).
\end{aligned} \label{Eq:St:UB:a}
\end{equation}
Since $G_{t_0} \geq 0$  and $G'_{t_0} = S_{t_0}$ from \eqref{Eq:RandProc2:Def}, combining \eqref{Eq:St:UB} and \eqref{Eq:St:UB:a} proves the inequality in the lemma statement for the case of $t\notin \mathcal{T}$, completing the proof. \hfill $\blacksquare$

\section{Proof of Lemma~\ref{Lem:St:Cross}}\label{App:St:Cross} 
 Using Lemma~\ref{Lem:St:UB} and for  $0 \leq a   \leq x$, 
\begin{align}
\Pr(S_t  > x) &\leq \Pr(G_t + G'_t > x)  \nn\\
&= \Pr(G_t >  x - G'_t\mid G'_t\geq a)\Pr\l(G'_t\geq a \r) +  \nn\\
&\qquad \qquad \Pr(G_t >  x - G'_t\mid G'_t <  a)\Pr\l(G'_t<  a \r)\nn\\
&\leq \Pr\l(G'_t\geq a \r) +  \Pr(G_t >  x - a). \label{Eq:St:Cross:a}
\end{align}
Let $t_0 \in \mathcal{T}$ specify the slot such that $G'_t = S_{t_0}$. Since $S_{t_0-1} < P$ based on the definition of $\mathcal{T}$, $S_{t_0} = S_{t_0-1} + Z_{t_0}$ using \eqref{Eq:Battery:Evol:Infinite}. It follows from this equality and \eqref{Eq:St:Cross:a} that 
\begin{align}
\Pr(S_t  > x)&\leq \Pr\l(S_{t_0-1} + Z_{t_0} \geq a \r) +  \Pr(G_t >  x - a)\nn\\
&\leq \Pr\l(Z_{t_0} \geq a - P \r) +  \Pr(G_t >  x - a)\nn\\
&= \Pr\l(\bar{Z}_{t_0} \geq a - 2P \r) +  \Pr(G_t >  x - a). \label{Eq:St:Cross:b}
\end{align}
By bounding the first term in \eqref{Eq:St:Cross:b} using Chernoff bound \cite{GallagerBook:StochasticProcs:95} and the second using \eqref{Eq:Gt:Cross}, 
\begin{align}
\Pr(S_t  > x)&\leq  \min_{r\geq 0} \E\l[e^{r\bar{Z}_1}\r]e^{-r(a-2P)}+  e^{-r^*(P)(x-a)}\nn\\
&\leq  e^{-r^*(P)(a-2P)}+  e^{-r^*(P)(x-a)} \label{Eq:St:Cross:c}
\end{align}
where \eqref{Eq:St:Cross:c} results from setting   $r= r^*(P)$. 
By choosing $a$ such that the exponents of the two terms in  \eqref{Eq:St:Cross:c}  are equal, the desired result follows.\hfill $\blacksquare$

\section{Proof of Lemma~\ref{Lem:ExcessEnergy}}
\label{App:ExcessEnergy} From the definition in \eqref{Eq:EnergEx}, 
\begin{align}
D_t(x) & = \int_x^\infty y f_s(y, t) dy - x\Pr(S_t > x)\nn\\
& = \int_x^\infty \Pr(S_t > y) dy\label{Eq:D:Exp:a}\\
& \leq \int_x^\infty 2e^{-\frac{1}{2}r^*(P)(y - 2P)} dy \label{Eq:D:Exp:b}
\end{align}
where \eqref{Eq:D:Exp:a} and \eqref{Eq:D:Exp:b} are obtained using  integration by parts and Lemma~\ref{Lem:St:Cross}, respectively. The desired result follows from \eqref{Eq:D:Exp:b}.\hfill $\blacksquare$ 

\section{Proof of Lemma~\ref{Lem:ExcessEnergy:B}}
 \label{App:ExcessEnergy:B} 
 
The definitions of the random processes  $\{G_t\}$ and $\{G'_t\}$ in \eqref{Eq:RandProc1:Def} and \eqref{Eq:RandProc2:Def} are modified for the case of finite-battery capacity. Specifically, $\{G_t\}$ is redefined as 
 \begin{equation}\label{Eq:RandProc1:Def:Finite}
G_t = \min(\max(G_{t-1} + \bar{Z}_t, 0), B-P)
\end{equation}
and $\{G'\}$ is as given in  \eqref{Eq:RandProc2:Def} but with  the battery-level evolution following  \eqref{Eq:Battery:Evol}.  Given finite battery capacity $B$, the inequality $S_t \leq G'_t + G_t$ can be proved using induction as follows. This inequality holds for $t=0$ since $S_0 = G_0 = 0$ and $G'_0 =P$. Assume that $S_t \leq G'_t + G_t$.  Consider the case of $(t+1) \in \mathcal{T}$. It follows from the definition of $\{G'_t\}$ in \eqref{Eq:RandProc2:Def} that  $S_{t+1}=G'_{t+1}$. Therefore,  $S_{t+1} \leq G'_{t+1} + G_{t+1}$ since $G_{t+1} \geq 0$ from \eqref{Eq:RandProc1:Def:Finite}. 
Next, consider the case of  $(t+1) \notin \mathcal{T}$. Based on the evolution of $\{S_t\}$ in  \eqref{Eq:Battery:Evol}, 
\begin{equation}
S_{t+1}  = \l\{ \begin{aligned}
&\min(S_t + Z_{t+1}, B), && S_t < P \\
&\min(S_t + \bar{Z}_{t+1}, B), && S_t \geq P. 
\end{aligned}\r.\label{Eq:S:UB:B:b}
\end{equation}
Given  that $(t+1) \notin \mathcal{T}$ and $S_t < P$, $S_t + Z_{t+1}\leq P$ based on the definition of $\mathcal{T}$. As a result, 
\begin{equation}
\min(S_t + Z_{t+1}, B)\leq G'_{t+1} + G_{t+1}\label{Eq:S:UB:B:d}
\end{equation}
since $G_{t+1}\geq 0 $ and $G'_{t+1}\geq P$ from  \eqref{Eq:RandProc2:Def}. If $S_t \geq P$, since $S_t \leq G'_t + G_t$, 
\begin{align}
\min(S_t + \bar{Z}_{t+1}, B) &\leq \min(G'_t + G_t + \bar{Z}_{t+1}, B)\nn\\
&\leq G'_t + \min(G_t + \bar{Z}_{t+1}, B-P)\label{Eq:S:UB:B:b:b}\\
&\leq G'_{t+1} + G_{t+1}\label{Eq:S:UB:B:b:c}
\end{align}
where \eqref{Eq:S:UB:B:b:b} applies $G'_t \geq P$, and \eqref{Eq:S:UB:B:b:c} uses \eqref{Eq:RandProc1:Def:Finite} and $G'_{t+1} = G'_t$ given that $(t+1) \notin \mathcal{T}$.  Combining \eqref{Eq:S:UB:B:b}, \eqref{Eq:S:UB:B:d} and \eqref{Eq:S:UB:B:b:c} proves  that $S_{t+1}\leq G'_{t+1} + G_{t+1}$ if $S_t\leq G'_t + G_t$. It follows that 
$S_t\leq G'_t + G_t$ for all $t \geq 0$. Furthermore, it can be shown by expanding \eqref{Eq:RandProc1:Def:Finite} that $\Pr(G_t > x)$ is no larger than that for the case of infinite battery capacity.  Using these results and following the same  procedures as for proving Lemma~\ref{Lem:St:Cross} and \ref{Lem:ExcessEnergy}, it can be shown that \eqref{Eq:St:Cross} also holds for the case of finite battery capacity  and 
\begin{equation}
\tilde{D}_t(x)  \leq \frac{4}{r^*(P)} e^{-\frac{1}{2}r^*(P)(x - 2P)}, \qquad \forall \ t \geq 1. \nn
\end{equation}
The desired result follows by setting $x = B$. \hfill $\blacksquare$

\section{Proof of Lemma~\ref{Lem:BatteryTail:B}}
 \label{App:BatteryTail:B} 
Define the random process $\{Q_t\}$ such that 
\begin{equation}\label{Eq:Q:Def}
Q_t = \min(Q_{t-1} + \bar{Z}_t, B),\qquad t = 1, 2, \cdots
\end{equation}
with $Q_0 = 0$.  Comparing \eqref{Eq:Q:Def} and the evolution  of $S_t$ in \eqref{Eq:Battery:Evol} shows that $S_t \geq Q_t$. Therefore, given $x \in [0, B)$ 
\begin{equation}
\Pr(S_t < x) \leq \Pr(Q_t < x). \label{Eq:St:Tail:B}
\end{equation}
 By expanding \eqref{Eq:Q:Def} 
\begin{equation}
\begin{aligned}
Q_t = \min\Bigg(B, & B+\bar{Z}_t, B + \bar{Z}_t + \bar{Z}_{t-1}, \cdots, \\
& B + \sum_{m=2}^t\bar{Z}_m, \sum_{m=1}^t\bar{Z}_m\Bigg). 
\end{aligned}
\label{Eq:Q:Expand}
\end{equation}
For ease of notation, define 
\begin{equation}
\tilde{Q}_t =  B + \min\l(0, \bar{Z}_t, \bar{Z}_t + \bar{Z}_{t-1}, \cdots, \sum_{m=2}^t\bar{Z}_m\r).\nn
\end{equation}
 Then $Q_t = \min(\tilde{Q}_t, \sum_{m=1}^t\bar{Z}_m)$. It follows that
\begin{align}
\Pr(Q_t < x) &= \Pr\l( \sum_{m=1}^t\bar{Z}_m <  x\mid Q_t =  \sum_{m=1}^t\bar{Z}_m\r)\times \nn\\
&\qquad \Pr\l(Q_t =  \sum_{m=1}^t\bar{Z}_m\r) +\nn\\
&  \qquad \Pr\l(\tilde{Q}_t <  x\mid Q_t = \tilde{Q}_t\r)\Pr\l(Q_t = \tilde{Q}_t\r) \nn\\
&\leq \Pr\l(Q_t =  \sum_{m=1}^t\bar{Z}_m\r) +\nn\\
& \qquad \Pr\l(\tilde{Q}_t <  x\mid Q_t = \tilde{Q}_t\r). \label{Eq:Q:UB}
\end{align}
By inspecting \eqref{Eq:Q:Expand}, the event $Q_t =  \sum_{m=1}^t\bar{Z}_m$ is equivalent to the one $\max(\bar{Z}_1, \bar{Z}_1+\bar{Z}_2, \cdots,  \sum_{m=1}^t\bar{Z}_m) \leq B$.  Then the inequality in \eqref{Eq:Q:UB} can be rewritten as 
\begin{align}
\Pr(Q_t < x)  &\leq \Pr\!\l(\!\max\!\l(\bar{Z}_1, \bar{Z}_1\!+\!Z_2, \cdots,  \sum_{m=1}^t\!\bar{Z}_m\!\r)\! \leq\! B\r) \!+\!\nn\\
\Pr\Bigg(\tilde{Q}_t& <  x\mid \max\!\l(\bar{Z}_1, \bar{Z}_1\!+\!Z_2, \cdots,  \sum_{m=1}^t\!\bar{Z}_m\r) \!>\!  B\Bigg).\nn
\end{align}
Note that removing  the conditioning of the last term increases the  probability. Therefore 
\begin{align}
\Pr(Q_t < x)  &\leq \Pr\!\l(\!\max\!\l(\bar{Z}_1, \bar{Z}_1\!+\!Z_2, \cdots,  \sum_{m=1}^t\!\bar{Z}_m\r) \!\leq\! B\!\r) +\nn\\
& \qquad  \Pr\l(\tilde{Q}_t <  x\r)\nn\\
&\leq \Pr\l(\sum_{m=1}^t\bar{Z}_m \leq B\r) + \Pr\l(\tilde{Q}_t <  x\r). \label{Eq:Q:UB:a}
\end{align}
Applying a similar technique as for proving the result in Theorem~\ref{Theo:TXProb:InfBattery} for the case of $\lambda_e > P$ shows that given $\lambda_e > P$ 
\begin{equation} \label{Eq:LLN:a}
\lim_{n\rightarrow\infty}\frac{1}{n}\sum_{t=1}^n \Pr\l(\sum_{m=1}^t\bar{Z}_m \leq B\r) = 0. 
\end{equation}
Using the definition of $\tilde{Q}_t$, the last term in \eqref{Eq:Q:UB:a} can be  rewritten as 
\begin{align}
\Pr\l(\tilde{Q}_t <  x\r)  &= \Pr\Bigg(\!\min\!\Bigg(\!0, \bar{Z}_t, \bar{Z}_t + \bar{Z}_{t-1}, \cdots, \sum_{m=2}^t\bar{Z}_m\!\Bigg)\leq \nn\\
&\qquad \qquad x - B\Bigg)\nn\\
&=  \Pr\Bigg(\!\min\!\Bigg(\bar{Z}_t, \bar{Z}_t + \bar{Z}_{t-1}, \cdots, \sum_{m=2}^t\bar{Z}_m\!\Bigg) \leq\nn\\
&\qquad \qquad x - B\Bigg)\nn
\end{align}
since $(x - B) \leq  0$.  Applying Kingman bound in a similar way as for obtaining \eqref{Eq:Prob:G} yields 
\begin{equation}\label{Eq:Prob:S:UB}
\Pr\l(\tilde{Q}_t <  x\r) \leq e^{r^*(P)(B-x)}
\end{equation}
where $r^*(P) < 0$ according to Assumption~\ref{AS:Z:Dist} given $\lambda_e > P$. By combining \eqref{Eq:Q:UB:a}, \eqref{Eq:LLN:a} and \eqref{Eq:Prob:S:UB}
\begin{equation}\label{Eq:Q:Lim}
\lim_{n\rightarrow\infty}\frac{1}{n}\sum_{t=1}^n\Pr(Q_t < x) \leq e^{r^*(P)(B-x)}. 
\end{equation}
The desired result follows from \eqref{Eq:St:Tail:B} and  \eqref{Eq:Q:Lim}. \hfill $\blacksquare$

\bibliographystyle{ieeetr}

\begin{IEEEbiographynophoto}{Kaibin Huang}  (S'05, M'08) received the B.Eng. (first-class hons.) and the M.Eng. from the National University of Singapore in 1998 and 2000, respectively, and the Ph.D. degree from The University of Texas at Austin (UT Austin) in 2008, all in electrical engineering.

Since Jul. 2012, he has been an assistant professor in the Dept. of Applied Mathematics at The Hong Kong Polytechnic University (PolyU), Hong Kong. He had held the same position in the School of Electrical and Electronic Engineering at Yonsei University, S. Korea from Mar. 2009 to Jun. 2012 and presently is affiliated with the school as an adjunct professor. From Jun. 2008 to Feb. 2009, he was a Postdoctoral Research Fellow in the Department of Electrical and Computer Engineering at the Hong Kong University of Science and Technology. From Nov. 1999 to Jul. 2004, he was an Associate Scientist at the Institute for Infocomm Research in Singapore. He frequently serves on the technical program committees of major IEEE conferences in wireless communications. He will chair the Comm. Theory Symp. of IEEE GLOBECOM 2014 and has been the technical co-chair for IEEE CTW 2013, the track chair for IEEE Asilomar 2011, and the track co-chair for IEE VTC Spring 2013 and IEEE WCNC 2011. He is an editor for the IEEE Wireless Communications Letters and also the Journal of Communication and Networks. He is an elected member of the SPCOM Technical Committee of the IEEE Signal Processing Society. Dr. Huang received the Outstanding Teaching Award from Yonsei, Motorola Partnerships in Research Grant, the University Continuing Fellowship at UT Austin, and a Best Paper Award from IEEE GLOBECOM 2006. His research interests focus on the analysis and design of wireless networks using stochastic geometry and multi-antenna limited feedback techniques.
\end{IEEEbiographynophoto}
\vfill
\end{document}

%% file: header.tex
\newtheorem{theorem}{Theorem}

\newtheorem{lemma}{Lemma}
\newtheorem{corollary}{Corollary}

\newtheorem{proposition}{Proposition}

\newtheorem{assumption}{Assumption}

\newtheorem{remark}{Remark}

\def\E{\mathsf{E}}

\def\phi{\varphi}

\def\SINR{\mathsf{SINR}}

\def\l{\left}
\def\r{\right}
\def\({\left(}
\def\){\right)}

\def\b0{{\mathbf{0}}}

\def\bP{{\mathbf{P}}}

\newcommand{\Pout}{P_{\mathsf{out}}}

\newcommand{\nn}{\nonumber}